\documentclass[11pt]{article}
\RequirePackage[colorlinks,citecolor=blue,urlcolor=blue,linkcolor=blue]{hyperref}
\usepackage{url}

\usepackage{tikz,pgf,pgfplots}
\usetikzlibrary{pgfplots.groupplots}
\usepackage{textpos}
\usetikzlibrary{calc,shadows}

\DeclareSymbolFont{slenderlargesymbols}{OMX}{ccex}{m}{n}
\DeclareMathSymbol{\prod}{\mathop}{slenderlargesymbols}{"51}
\usepackage{latexsym,amssymb,amsmath,amsthm,graphics,graphicx,float,psfrag, epsfig, color, authblk, enumerate,enumitem}
\usepackage{mathabx}
\newtheorem{theorem}{Theorem} 
\newtheorem{lemma}[theorem]{Lemma}
\newtheorem{proposition}[theorem]{Proposition}

\usepackage{epstopdf} 
\usepackage{algorithm} 
\usepackage{algorithmic} 

\newcommand{\inner}[3][]{{\langle #2,#3 \rangle_{#1}}} 

\newcommand{\gtilde}{\tilde{g}}

\newcommand{\R}{\mathbb{R}}

\newcommand{\eps}{\epsilon}
\newcommand{\rhotilde}{\tilde{\rho}}

\newcommand{\thetabar}{\overline{\theta}}
\newcommand{\thetacheck}{\widecheck{\theta}}
\newcommand{\zetacheck}{\rho}

\newcommand{\xo}{x_\ast}

\newcommand{\Wplus}[1]{W_{+, #1}}

\newcommand{\Wpx}{\Wplus{x}}

\newcommand{\E}{\mathbb{E}}

\newcommand{\Mxyhat}{M_{\hat{x} \leftrightarrow \hat{y}}}

\newcommand{\xhat}{\hat{x}}
\newcommand{\yhat}{\hat{y}}
\newcommand{\xohat}{\hat{x}_0}

\newcommand{\Qxy}{Q_{x,y}}

\newcommand{\PiWd}{\Pi_{i=d}^1 W_{i,+,x}}

\newcommand{\PiWdo}{\Pi_{i=d}^1 W_{i,+,\xo}}

\newcommand{\Pipithetapii}{\prod_{i=0}^{d-1} \frac{\pi - \thetabar_i}{\pi}}
\newcommand{\zetaterm}{\sum_{i=0}^{d-1} \frac{\sin \thetabar_i}{\pi} \prod_{j=i+1}^{d-1} \frac{\pi - \thetabar_j}{\pi}}

\newcommand{\Wnpx}[1]{W_{#1, +, x}}

\newcommand{\Wipx}{W_{i, +, x}}

\newcommand{\Win}[1]{W_{#1}}
\newcommand{\Wi}{\Win{i}}

\newcommand{\vbarx}{\overline{v}_{x}}

\newcommand{\hx}{h_x}

\newcommand{\feta}{f_\eta}
\newcommand{\fzero}{f_0}
\newcommand{\etabar}{\overline{\eta}}

\DeclareMathOperator{\Span}{span}
\DeclareMathOperator{\relu}{relu}
\DeclareMathOperator{\diag}{diag}

\DeclareMathOperator{\conv}{conv}

\usepackage{ifthen}
\usepackage{mathtools}
\newcommand\defeq{\coloneqq}
\newcommand{\noise}{\eta}

\newcommand\norm[2][\Tnorm]{\ensuremath{{\left\|#2\right\|}_{#1}}}

\newcommand\Tinnerprod{}
\newcommand{\innerprod}[3][\Tinnerprod]{\ifthenelse{\equal{#1}{}}{\ensuremath{\left<#2,#3\right>}}{\ensuremath{\left<#2,#3\right>_{#1}}}}

\newcommand\Tex{}
\newcommand\PR[2][\Tex]{
\ifthenelse{\equal{#1}{}}{{\mathbb P}\left[#2\right]}{\ensuremath{{\mathbb P}_{#1}\left[ #2\right]}}}
\newcommand\EX[2][\Tex]{
\ifthenelse{\equal{#1}{}}{{\mathbb E}\left[#2\right]}{\ensuremath{{\mathbb E}_{#1}\left[ #2\right]}}}

\newcommand\setS{\mathcal S}
\newcommand\qbarx{\bar q_x}
\newcommand\eventnoisesmall{\mathcal E_{\text{noise}}}

\newcommand\stepdir[1]{\tilde v_{#1}}
\newcommand\ball{\mathcal B}
\newcommand{\MSE}{\mathrm{MSE}}

\newcommand{\PiWdi}{\Pi_{i=d}^1 W_{i,+,x_i}}
\newcommand\reals{\mathbb R}
\newcommand{\PiWditilde}{\Pi_{i=d}^1 W_{i,+,\tilde{x}_i}}
\newcommand{\PS}{P_S}

\newcommand{\whcomm}[2]{{\color{black} #1}{\color{black} #2}} 

\usepackage[
backend=bibtex,
url=false,isbn=false,doi=false,
date=year,
hyperref=auto,
style=alphabetic,
natbib=true,
sorting=nty,
firstinits=true,
maxnames=2, maxbibnames=10,
]{biblatex}

\bibliography{./bibliography.bib} 

\title{Deep Denoising}

\author{ Reinhard Heckel\thanks{Department of Electrical and Computational Engineering, Rice Univeristy, Houston, TX}, Wen Huang and Paul Hand\ \thanks{Department of Computational and Applied Mathematics, Rice University, Houston, TX} \  and Vladislav Voroninski\thanks{Helm.ai, Menlo Park, CA}}

\begin{document}

\begin{center}

{\bf{\LARGE{
Rate-Optimal Denoising with Deep Neural Networks
}}}

\vspace*{.2in}

{\large{
\begin{tabular}{cccc}
Reinhard Heckel$^{\ast}$ & Wen Huang$^{\star}$ & Paul Hand$^{\star\star}$ & Vladislav Voroninski$^{\dagger}$\\
\end{tabular}
}}

\vspace*{.05in}

\begin{tabular}{c}
Department of Electrical and Computer Engineering$^\ast$, Rice University \\
School of Mathematical Sciences$\star$, 
Xiamen University, Xiamen, Fujian, P.R.China\\ 
Department of Mathematics and College of Computer and Information Science$^{\star\star}$, Northeastern University\\
Helm.ai$^\dagger$
\end{tabular}

\vspace*{.1in}

\today

\vspace*{.1in}

\end{center}

\begin{abstract}
Deep neural networks provide state-of-the-art performance for image denoising, where the goal is to recover a near noise-free image from a noisy observation.
The underlying principle is that neural networks trained on large datasets have empirically been shown to be able to generate natural images well from a low-dimensional latent representation of the image.
Given such a generator network, a noisy image can be denoised by i) finding the closest image in the range of the generator or by ii) passing it through an encoder-generator architecture (known as an autoencoder).
However, there is little theory to justify this success, let alone to predict the denoising performance as a function of the network parameters.
In this paper we consider the problem of denoising an image from additive Gaussian noise using the two generator based approaches. 
In both cases, we assume the image is well described by a deep neural network with ReLU activations functions, mapping a $k$-dimensional code to an $n$-dimensional image. In the case of the autoencoder, we show that the feedforward network reduces noise energy by a factor of $O(k/n)$. 
In the case of optimizing over the range of a generative model, we state and analyze a simple gradient algorithm that minimizes a non-convex loss function, and provably reduces noise energy by a factor of $O(k/n)$.
We also demonstrate in numerical experiments that this denoising performance is, indeed, achieved by generative priors learned from data.
\end{abstract}

\section{Introduction} \label{sec:introduction}
We consider the denoising problem, where the goal is to remove noise from an unknown image or signal.
In more detail, our goal is to obtain an estimate of an image or signal $y_\ast \in \R^n$ from a noisy measurement
\[
y = y_\ast + \noise.
\]
Here, $\noise$ is unknown noise, which we model as a zero-mean white Gaussian random variable with covariance matrix $\sigma^2/n I$.
Image denoising relies on generative or prior assumptions on the image $y_\ast$, 
such as self-similarity within images~\citep{dabov_image_2007},
 sparsity in fixed~\citep{donoho_de-noising_1995} and learned bases~\citep{elad_image_2006}, and most recently, by assuming the image can be generated by a pre-trained deep-neural network~\citep{burger_image_2012,zhang_beyond_2017}.
Deep-network based approaches, typically yield the best denoising performance. This success can be attributed to their ability to efficiently represent and learn realistic image priors, for example via auto-decoders~\citep{hinton_reducing_2006} and generative adversarial models~\citep{goodfellow_generative_2014}.

Motivated by this success story, we assume that the image $y_\ast$ lies in the range of an image-generating network. 
In this paper, we propose the first algorithm for solving denoising with deep generative priors that provably finds an approximation of the underlying image. 
As the influence of deep networks in denoising and inverse problems grows, it becomes increasingly important to understand their performance at a theoretical level.  Given that most optimization approaches for deep learning are first-order gradient methods, a justification is needed for why they do not get stuck in local minima.  

The most related work that establishes theoretical reasons for why gradient methods might succeed when using deep generative priors for solving inverse problems, is \citep{hand_global_2017}.  In it, the authors establish global favorability for optimization of a $\ell_2$-loss function under a random neural network model.  Specifically, they show existence of a descent direction outside a ball around the global optimizer and a negative multiple of it in the latent space of the generative model.  This work does not justify why the one spurious point is avoided by gradient descent, nor does it provide a specific algorithm which provably estimates the global minimizer, nor does it provide an analysis of the robustness of the problem with respect to noise.  This work was subsequently extended to include the case of generative convolutional neural networks by \cite{ma2018invertibility}, but that work too does not prove convergence of a specific algorithm.




\paragraph{Contributions:}
The goal of this paper is to analytically quantify the denoising performance of deep-prior based denoisers.
Specifically, we characterize the performance of two simple and efficient algorithms for denoising based on a $d$-layer generative neural network
$G\colon \R^k \to \R^{n}$, with $k<n$.

We first provide a simple result for an encoder-generator network $G( E(y) )$ where $E \colon \R^n \to \R^k$ is an encoder network.
We show that if we pass noise through an encoder-decoder network $G( E(y) )$ that acts as the identity on a class of images of interest, then it reduces the random noise by $O(k/n)$. 
This result requires no assumptions on the weights of the network.

The second and main result of our paper pertains to denoising by optimizing over the latent code of a generator network with random weights.
We propose a gradient method 
that attempts to minimize the least-squares loss $f(x) = \frac{1}{2} \norm{G(x) - y}^2$ between the noisy image $y$ and an image in the range of the generator, $G(x)$. Even though $f$ is non-convex, we show that a gradient method yields an estimate $\hat x$ obeying
\begin{align*}
\norm{G(\hat x) - y_\ast}^2
\lesssim \sigma^2 \frac{k}{n},
\end{align*}
with high probability, where the notation $\lesssim$ absorbs a constant factor depending on the number of layers of the network, and its expansitivity, as discussed in more detail later.
Our result shows that the denoising rate of a deep generator based denoiser is optimal in terms of the dimension of the latent representation.
We also show in numerical experiments, that this rate---shown to be analytically achieved for random priors---is also experimentally achieved for priors learned from real imaging data.

\paragraph{Related work:}
We hasten to add that a close theoretical work to the question considered in this paper is the paper~\citep{bora_compressed_2017}, which solves a noisy compressive sensing problem with generative priors by minimizing an $\ell_2$-loss. Under the assumption that the network is Lipschitz, they show that if the global optimizer can be found, which is in principle NP-hard, then a signal estimate is recovered to within the noise level.  While the Lipschitzness assumption is quite mild, the resulting theory does not provide justification for why global optimality can be reached.


\section{Background on denoising with classical and deep-learning based methods}

As mentioned before, image denoising relies on modeling or prior assumptions on the image $y_\ast$. For example, suppose that the image $y_\ast$ lies in a $k$-dimensional subspace of $\mathbb R^n$ denoted by $\mathcal Y$.
Then we can estimate the original image by finding the closest point in $\ell_2$-distance to the noisy observation $y$ on the subspace $\mathcal Y$.
The corresponding estimate, denoted by $\hat y$, obeys
\begin{align}
\label{eq:subspacerate}
\norm{\hat y - y_\ast}^2 \lesssim \sigma^2 \frac{k}{n},
\end{align}
with high probability (throughout, $\norm{\cdot}$ denotes the $\ell_2$-norm).
Thus, the noise energy is reduced by a factor of $k/n$ over the trivial estimate $\hat y = y$ which does not use any prior knowledge of the signal.
The denoising rate~\eqref{eq:subspacerate} shows that the more concise the image prior or image representation (i.e., the smaller $k$), the more noise can be removed. If on the other hand the image model (the subspace, in this example) does not include the original image $y_\ast$, then the error bound~\eqref{eq:subspacerate} increases, as we would remove a significant part of the signal along with noise when projecting onto the range of the image prior.
Thus a concise and accurate model is crucial for denoising.

Real world signals rarely lie in \emph{a priori} known subspaces, and the last few decades of image denoising research have developed sophisticated algorithms based on accurate image models.
Examples include algorithms based on sparse representations in overcomplete dictionaries such as wavelets~\citep{donoho_de-noising_1995} and curvelets~\citep{starck_curvelet_2002}, and algorithms based on exploiting self-similarity within images~\citep{dabov_image_2007}.
A prominent example of the former class of algorithms is the (state-of-the-art) BM3D~\citep{dabov_image_2007} algorithm. 
However, the nuances of real world images are difficult to describe with handcrafted models. Thus, starting with the paper~\citep{elad_image_2006} that proposes to learn sparse representation based on training data,
it has become common to learn concise representation for denoising (and other inverse problems) from a set of training images.

\citet{burger_image_2012} applied deep networks to the denoising problem, by training a plain neural network on a large set of images.
Since then, deep learning based denoisers~\citep{zhang_beyond_2017} have set the standard for denoising. The success of deep network priors can be attributed to their ability to efficiently represent and learn realistic image priors, for example via auto-decoders~\citep{hinton_reducing_2006} and generative adversarial models~\citep{goodfellow_generative_2014}.
Over the last few years, the quality of deep priors has significantly improved \citep{Karras17, DIP,heckel_deep_2019}. As this field matures, priors will be developed with even smaller latent code dimensionality and more accurate approximation of natural signal manifolds. Consequently, the representation error from deep priors will decrease, and thereby enable even more powerful denoisers.


\section{Denoising with a neural network with an hourglass structure}

Perhaps the most straight-forward and classical approach to using deep networks for denoising is to train a deep network with an autoencoder or hourglass structure end-to-end to perform denoising.   
An autoencoder compresses data from the input layer into a low-dimensional code, and then generates an output from that code. In this section, we analyze such networks from the perspective of denoising.

Specifically, we show mathematically that a simple model for neural networks with an hourglass structure achieves optimal denoising rates, as given by the dimensionality of the low-dimensional code.
An autoencoder $H(x) = G( E(x) )$ 
consists of an encoder network $E \colon \reals^n \to \reals^k$ mapping an image to a low-dimensional latent representation, and a decoder or generator network 
$G \colon \reals^k \to \reals^n$
mapping the latent representation to an image. 
To see that the size of the low-dimensional code, $k$, determines the denoising rate, consider a simple one-layer encoder and multilayer decoder of the form
\begin{align}
\label{eq:encoderdecoder}
E(y)
=
\relu(W'y)
,
\quad
G(x) = \relu(\Win{d} \ldots \relu (\Win{2} \relu(\Win{1} x)) \ldots ),
\end{align}
where $\relu(x) = \max(x, 0)$ applies entrywise,
$W'$ are the weights of the encoder, 
$W_i \in \R^{n_i \times n_{i-1}}$, are the weights in the $i$-th layer of the decoder, and $n_i$ is the number of neurons in the $i$th layer.

Typically, the autoencoder network $H$ is trained such that 
$H(y) \approx y$ for some class of signals of interest (say, natural images).
The following proposition shows that the structure of the network alone guarantees that an autoencoder ``filters out'' most of the noise. 

\begin{proposition}
\label{prop:autodec}
Let $H = G \circ E$ be an autoencoder of the form~\eqref{eq:encoderdecoder} and note that it is piecewise linear, i.e., $H(y) = Uy$ in a region around $y$. 
Suppose that $\norm{U}^2 \leq  2$ for all such  regions. 
Let $\noise$ be Gaussian noise with covariance matrix $\sigma I$, $\sigma > 0$. Then, provided that 
$ k\cdot 32 \log(2n_1 n_2 \ldots n_d)   \leq n$, we have 
with probability at least $1 - 2 e^{-k \log(2n_1 n_2 \ldots n_d)}$, that 
\[
\norm[2]{H(\eta)}^2 \leq  
5\frac{k}{n} \log(2n_1 n_2 \ldots n_d) \norm[2]{\eta}^2.
\]
\end{proposition}

Note that the assumption $\norm{U}^2 \leq 2$ implies that
$\norm[2]{H(y)}^2 / \norm[2]{y}^2 \leq 2$ for all $y$. This in turn guarantees that the autoencoder does not ``amplify'' a signal too much.
 Note that we envision an autoencoder that is trained such that it obeys $H(y) \approx y$ for $y$ in a class of images. 
 The proposition would then justify why the the feedforward network reduces noise by $O(k/n)$. 

The proof of Proposition~\ref{prop:autodec}, contained in the appendix, shows that $H$ lies in the range of a union of $k$-dimensional subspaces and then uses a standard concentration argument showing that the union of those subspaces can represent no more than a fraction of $O(k/n)$ of the noise.

In the remainder of the paper we shows that denoising via enforcing a generative prior gives us an analogous denoising rate.


\section{\label{sec:erm}Denoising via enforcing a generative model}

We consider the problem of estimating a vector $y_\ast \in \R^{n}$
from a noisy observation $y = y_\ast + \noise$.
We assume that the vector $y_\ast$ belongs to the range of a $d$-layer generative neural network $G\colon \R^k \to \R^{n}$, with $k<n$.  That is, $y_\ast = G(x_\ast)$ for some $x_\ast \in \R^k$.
We consider a generative network of the form
\[
G(x) = \relu(\Win{d} \ldots \relu (\Win{2} \relu(\Win{1} \xo)) \ldots ),
\]
where $\relu(x) = \max(x, 0)$ applies entrywise,
$W_i \in \R^{n_i \times n_{i-1}}$, are the weights in the $i$-th layer,
$n_i$ is the number of neurons in the $i$th layer,
and the network is expansive in the sense that $k=n_0 < n_1 < \cdots < n_d = n$.
The problem at hand is:
Given the weights of the network $W_1\ldots W_d$ and a noisy observation $y$, obtain an estimate $\hat y$ of the original image $y_\ast$ such that $\norm{\hat y - y_\ast}$ is small and $\hat y$ is in the range of $G$.

\subsection{Enforcing a generative model}

As a way to solve the above problem, we first obtain an estimate of $\xo$, denoted by $\hat x$, and then estimate $y_\ast$ as $G(\hat x)$.
In order to estimate $\xo$, we minimize the loss
\begin{align*}
f(x) := \frac{1}{2} \norm{ G(x) - y }^2. 
\end{align*}
Since this objective is nonconvex, there is no \textit{a priori} guarantee of efficiently finding the global minimum. 
Approaches such as gradient methods could in principle get stuck in local minima, instead of finding a global minimizer that is close to $\xo$.

However, as we show in this paper, under appropriate conditions, a gradient method ---introduced next---finds a point that is very close to the original latent parameter $\xo$, with the distance to the parameter $\xo$ controlled by the noise.
In order to state the algorithm, we first introduce a useful quantity.
For analyzing which rows of a matrix $W$ are active when computing $\relu(Wx)$, we let
\[
\Wpx = \diag(Wx>0) W.
\]
For a fixed weight matrix $W$, the matrix $\Wpx$ zeros out the rows of $W$ that do not have a positive dot product with $x$.  Alternatively put, $\Wpx$ contains weights from only the neurons that are active for the input $x$.  We also define $\Wnpx{1} = (W_1)_{+,x} =  \diag(W_{1}x>0) W_1$ and
\begin{align*}
\Wipx = \diag(\Wi \Wnpx{i-1} \cdots \Wnpx{2} \Wnpx{1} x > 0) \Wi.
\end{align*}
The matrix $\Wipx$ consists only of the weights of the neurons in the $i$th layer that are active if the input to the first layer is  $x$.


\begin{figure}
\centering
\begin{tikzpicture}[scale=1]
\begin{groupplot}[
group style={group size=2 by 1,horizontal sep=1.4cm},
xlabel=$x_1$,
ylabel=$x_2$,
zlabel=$f(x)$,
width=0.5\textwidth,/tikz/font=\small,colormap/cool]
    \nextgroupplot[]
	 \addplot3[surf,shader=faceted] table[x index=1,y index=0,z index=2] {./fig/surf.dat};
  \end{groupplot}
\end{tikzpicture}
\caption{\label{fig:losssurface} Loss surface $f(x) = \norm{G(x) - G(\xo)}$, $\xo=[1,0]$, of an expansive network $G$ with ReLU activation functions with $k=2$ nodes  in the input layer and $n_2=300$ and  $n_3=784$ nodes in the hidden and output layers, respectively, with random Gaussian weights in each layer.
The surface has a critical point near $-\xo$, a global minimum at $\xo$, and a local maximum at $0$.
}
\end{figure}
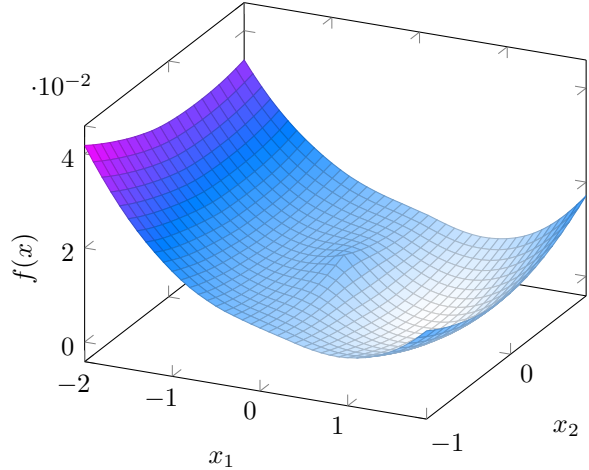%

We are now ready to state our algorithm: a gradient method with a tweak informed by the loss surface of the function to be minimized.
Given a noisy observation $y$, the algorithm  starts with an arbitrary initial point $x_0 \neq 0$.
At each iteration $i=0,1,\ldots$, the algorithm  computes the step direction
\[
\stepdir{x_i} = (\PiWdi)^t (G(x_i) - y),
\]
which is equal to the gradient of $f$ if $f$ is differentiable at $x_i$.
It then takes a small step opposite to $\stepdir{x_i}$.
The tweak is that before each iteration, the algorithm checks whether $f(-x_i)$ is smaller than $f(x_i)$, and if so, negates the sign of the current iterate $x_i$.

This tweak is informed by the loss surface.
To understand this step, it is instructive to examine the loss surface for the \emph{noiseless} case in Figure~\ref{fig:losssurface}.
It can be seen that while the loss function has a \emph{global} minimum at $\xo$, it is relatively flat  close to $-\xo$.  In expectation, there is a critical point that is a negative multiple of $\xo$ with the property that the curvature in the $\pm \xo$ direction is positive, and the curvature in the orthogonal directions is zero. Further, around approximately $-\xo$, the loss function is larger than around the optimum $\xo$.
As a simple gradient descent method (without the tweak) could potentially get stuck in this region, the negation check provides a way to avoid converging to this region.
Our algorithm is formally summarized as Algorithm~\ref{alg} below.

\begin{algorithm}[h!]
\caption{Gradient method}
\label{alg}
\begin{algorithmic}[1]
\REQUIRE Weights of the network $W_i$, noisy observation $y$, and step size $\alpha > 0$
\STATE Choose an arbitrary initial point $x_0 \in \mathbb{R}^k\setminus \{0\}$
\FOR {$i = 0, 1, \ldots$}
\IF {$f(-x_{i}) < f(x_{i})$}  \label{alg:st3}
\STATE $\tilde{x}_{i} \gets - x_{i}$;
\ELSE
\STATE $\tilde{x}_i \gets x_i$;
\ENDIF
\STATE   Compute $v_{\tilde{x}_i} \in \partial f(\tilde{x}_i)$, in particular, if $G$ is differentiable at $\tilde{x}_i$, then set $v_{\tilde{x}_i} = \tilde{v}_{\tilde{x}_i}$,  where
$$
\tilde{v}_{\tilde{x}_i} := (\PiWditilde)^t (G(\tilde{x}_i) - y);
$$
\STATE $x_{i+1} \gets \tilde{x}_i - \alpha {v}_{\tilde{x}_i}$;
\ENDFOR
\end{algorithmic}
\end{algorithm}

Other variations of the tweak are also possible.  For example, the negation check in Step~\ref{alg:st3} could be performed after a convergence criterion is satisfied, and if a lower objective is achieved by negating the latent code, then the gradient descent can be continued again until a convergence criterion is again satisfied.

\section{Main results}

For our analysis, we consider a fully-connected generative network $G\colon \R^k \to \R^n$ with Gaussian weights and no bias terms. Specifically, we assume that the weights $W_i$ are independently and identically distributed as~$\mathcal N(0,2/n_i)$, but do not require them to be independent across layers.
Moreover, we assume that the network is sufficiently \textit{expansive}:

\newtheorem*{block}{Expansivity condition}
\begin{block} 
We say that the expansivity condition with constant $\epsilon>0$ holds if
\[
n_i \geq c \epsilon^{-2} \log(1/\eps) n_{i-1} \log n_{i-1},
\quad
\text{for all }
i,
\]
where $c$ is a particular numerical constant.
\end{block}
\noindent
In a real-world generative network the weights are learned from training data, and are not drawn from a Gaussian distribution. Nonetheless, the motivation for selecting Gaussian weights for our analysis is as follows:
\begin{enumerate}
\item The empirical distribution of weights from deep neural networks often have statistics consistent with Gaussians.  AlexNet is a concrete example~\citep{Arora15}.
\item The field of theoretical analysis of recovery guarantees for deep learning is nascent, and Gaussian networks can permit theoretical results because of well developed theories for random matrices.
\item It is not clear which non-Gaussian distribution for weights is superior from the joint perspective of realism and analytical tractability.
\end{enumerate}

The network model we consider is fully connected.  We anticipate that the analysis of this paper can be extended to the case of generative convolutional neural networks.  This extension has already happened for theoretical results concerning the favorability of the optimization landscape for compressive sensing under generative priors~\citep{ma2018invertibility}, as mentioned previously. 

\noindent We are now ready to state our main result. 


\begin{theorem}
\label{thm:main-rescaled}
Consider a network with the weights in the $i$-th layer, $W_i \in \R^{n_i \times n_{i-1}}$, i.i.d.~$\mathcal N(0,2/n_i)$ distributed, and suppose that the network satisfies the expansivity condition for $\epsilon = K/d^{90}$. 
Also, suppose that the noise variance $\omega$, defined for notational convenience as 
\[
\omega
\defeq
\sqrt{ 18\sigma^2 \frac{k}{n}
\log(n_1^d n_2^{d-1}\ldots n_d)},
\]
obeys
\[
\omega \leq \frac{\norm{\xo} K_1}{d^{16}}.
\]
Consider the iterates of Algorithm~\ref{alg} with stepsize
$\alpha = K_4\frac{1}{d^2}$.
Then, there exists a number of steps $N$ upper bounded by
\[
N
\leq
K_2 d^{86}
\frac{f(x_0)}{ \norm{\xo}}
\]
such that after $N$ steps, the iterates of Algorithm~\ref{alg} obey
\begin{align}
\label{eq:errorxixo-rescaled}
\norm{x_i -\xo}
\leq
\frac{K_5}{d^{36}}  \norm{\xo} 
+
K_6 d^6 \omega,
\quad
\text{for all $i \geq N$},
\end{align}
with probability at least
$1-2e^{-2k \log n} - \sum_{i=2}^d 8 n_i e^{-K_7 n_{i-2} }
- 8 n_1 e^{-K_7  k  \log d / d^{180}}
$.
In addition, for all $i > N$, we have
\begin{align}
\|x_{i} - x_*\| &\leq (1 -  \alpha 7/8 )^{i - N} \|x_N - x_*\| + K_8 \omega \hbox{ and }  \label{GA:e68} \\
\|G(x_{i}) - G(x_*)\| &\leq 1.2 (1 -  \alpha 7/8 )^{i - N} \|x_N - x_*\| + 1.2 K_8 \omega, \label{GA:e74}
\end{align}
where $\alpha < 1$ is the stepsize of the algorithm.
Here, $K_1,K_2, \ldots$ are numerical constants, and $x_0$ is the initial point in the optimization.
\end{theorem}


Our result guarantees that after polynomially many iterations (with respect to $d$), the algorithm converges linearly to a region satisfying
\[
\norm{x_i - \xo}^2
\lesssim
\sigma^2 \frac{k}{n},
\]
where the notation $\lesssim$ absorbs a factor logarithmic in $n$ and polynomial in $d$.
One can show that $G$ is Lipschitz in a region around $\xo$\footnote{The proof of Lipschitzness follows from applying the Weight Distribution Condition in Section~\ref{sec:WDC}.},
\[
\norm{G(x_i) - G(\xo)}^2
\lesssim
\sigma^2 \frac{k}{n}.
\]
Thus, the theorem guarantees that our algorithm yields the denoising rate of $\sigma^2 k/n$, and, as a consequence, denoising based on a generative deep prior provably reduces the energy of the noise in the original image by a factor of $k/n$.  

In the case of $\sigma=0$, the theorem guarantees convergence to the global minimizer $\xo$.   We note that the intention of this paper is to show rate-optimality of recovery with respect to the noise power, the latent code dimensionality, and the signal dimensionality.  As a result, no attempt was made to establish optimal bounds with respect to the scaling of constants or to powers of $d$.  The bounds provided in the theorem are highly conservative in the constants and dependency on the number of layers, $d$, in order to keep the proof as simple as possible.  Numerical experiments shown later reveal that the parameter range for successful denoising are much broader than the constants suggest. As this result is the first of its kind for rigorous analysis of denoising performance by deep generative networks, we anticipate the results can be improved in future research, as has happened for other problems, such as sparsity-based compressed sensing and phase retrieval.  

Finally, we remark that Theorem~\ref{thm:main-rescaled} can be generalized to the case where $y_\ast$ only approximately lies in the range of the generator, i.e., 
$G(x_\ast) \approx y_\ast$. Specifically, if $\norm{G(x_\ast) - y_\ast }$ is sufficiently small, the error induced by this perturbation is proportional to \norm{G(x_\ast) - y_\ast }.


\subsection{\label{sec:WDC}
The Weight Distribution Condition (WDC)}

To prove our main result, we make use of a deterministic condition on $G$, called the Weight Distribution Condition (WDC), and then show that Gaussian $W_i$, as given by the statement of Theorem~\ref{thm:main-rescaled} are such that $W_i / \sqrt{2}$ satisfies the WDC with the appropriate probability for all $i$, provided the expansivity condition holds.
Our main result, Theorem~\ref{thm:main-rescaled}, continues to hold for any weight matrices such that $W_i/\sqrt{2}$ satisfy the WDC.

The condition is on the spatial arrangement of the network weights within each layer.
We say that the matrix $W \in \R^{n \times k}$ satisfies the \textit{Weight Distribution Condition} with constant $\eps$ if for all nonzero $x,y \in \R^k$,
\begin{align}
\Bigl \| \sum_{i=1}^n 1_{\innerprod{w_i}{x}>0} 1_{\innerprod{w_i}{y}>0} \cdot w_i w_i^t  - \Qxy \Bigr \| \leq \eps,  \text{ with } \Qxy = \frac{\pi - \theta_0}{2 \pi} I_k + \frac{\sin \theta_0}{2\pi}  \Mxyhat, \label{WDC}
\end{align}
where $w_i \in \R^k$ is the $i$th row of $W$; $\Mxyhat \in \R^{k \times k}$ is the matrix\footnote{A formula for $\Mxyhat$ is as follows.  If $\theta_0 = \angle(\xhat, \yhat) \in (0, \pi)$ and $R$ is a rotation matrix such that $\xhat$ and $\yhat$ map to $e_1$ and $\cos \theta_0 \cdot e_1 + \sin \theta_0 \cdot e_2$ respectively, then $\Mxyhat = R^t \begin{pmatrix} \cos \theta_0 & \sin \theta_0 & 0 \\ \sin \theta_0 & - \cos \theta_0 & 0 \\ 0 & 0 & 0_{k-2} \end{pmatrix} R$, where $0_{k-2}$ is a $k-2 \times k-2$ matrix of zeros.  If $\theta_0 = 0$ or $\pi$, then $\Mxyhat = \xhat \xhat^t$ or $- \xhat \xhat^t$, respectively.} such that $\xhat \mapsto \yhat$, $\yhat \mapsto \xhat$, and $z \mapsto 0$ for all $z \in \Span(\{x,y\})^\perp$;  $\xhat = x/\|x\|_2$  and $\yhat = y /\|y\|_2$;  $\theta_0 = \angle(x, y)$; and $1_S$ is the indicator function on $S$.  The norm in the left hand side of \eqref{WDC} is the spectral norm.  Note that an elementary calculation\footnote{To do this calculation, take $x=e_1$ and $y = \cos \theta_0\cdot e_1 + \sin \theta_0 \cdot e_2$ without loss of generality.  Then each entry of the matrix can be determined analytically by an integral that factors in polar coordinates.} gives that $\Qxy = \E[\sum_{i=1}^n 1_{\innerprod{w_i}{x}>0} 1_{\innerprod{w_i}{y}>0} \cdot w_i w_i^t ]$ for $w_i \sim \mathcal{N}(0, I_k/n)$.  As the rows $w_i$ correspond to the neural network weights of the $i$th neuron in a layer given by $W$, the WDC provides a deterministic property under which the set of neuron weights within the layer given by $W$ are distributed approximately like a Gaussian.  The WDC could also be interpreted as a deterministic property under which the neuron weights are distributed approximately like a high dimensional vector chosen uniformly from a sphere of a particular radius.  Note that if $x=y$, $\Qxy$ is an isometry up to a factor of $1/2$.

\subsection{Sketch of proof of Theorem~\ref{thm:main-rescaled}}

The proof relies on a characterization of the loss surface. 
We show that outside of two balls around $x = \xo$ and $x = - \rho_d \xo$, with $\rho_d$ a constant defined in the proof, the direction chosen by the algorithm is a descent direction, with high probability.

We show that the stepdirection $\stepdir{x}$ concentrates around a particular $\hx \in \R^k$, that is a continuous function of nonzero $x,\xo$ and is zero only at $x = \xo$, $x = - \rho_d \xo$, and $0$, using a concentration argument similar to~\citep{hand_global_2017}.
Around $x = \xo$, the loss function has a global minimum, close to $0$ it has a saddle point, and close to $x = - \rho_d \xo$
potentially a local minimum. 
In a nutshell, we show that 
i) the algorithm moves away from the saddle point at $0$, and
ii) the algorithm escapes the local minimum close to $x = - \rho_d \xo$ with the twist in Steps 3-5 of the algorithm. 
Finally, the iterates end up close to the $\xo$.

The proof is organized as described next and as illustrated in Figure~\ref{sec:ProofLogic}. 
The algorithm is initialized at an arbitrary point; for example close to $0$. Algorithm~\ref{alg} moves away from $0$, at least till its iterates are outside the gray ring, as $0$ is a local maximum; and once an iterate $x_i$ leaves the gray ring around $0$, all subsequent iterates will never be in the white circle around $0$ again (see Lemma~\ref{lem:notinzeroball} in the supplement).
Then the algorithm might move towards $-\rho_d x_\ast$, but once it enters the dashed ball around $-\rho_d x_\ast$, it enters a region where the function value is strictly larger than that of the dashed ball around $x_\ast$, by Lemma~\ref{GA:le14} in the supplement.
Thus steps 3-5 of the algorithm will ensure that the next iterate $x_i$ is in the dashed ball around $x_\ast$.
From there, the iterates will move into the region $\mathcal S_\beta^+$, since outside of $\mathcal S_\beta^+ \cup \mathcal S_\beta^-$ the algorithm chooses a descent direction in each step (see the argument around equation~\eqref{GA:e21} in the supplement).
The region $\mathcal S_\beta^+$ is covered by a ball of radius $r$, by Lemma~\ref{lemma:Sepsst} in the supplement, determined by the noise and $\epsilon$.
This establishes the bound~\eqref{eq:errorxixo-rescaled} in the theorem. 

The proof proceeds be showing that within a ball around $\xo$, the algorithm then converges linearly, which establishes equations~\eqref{GA:e68} and \eqref{GA:e74}.

\usetikzlibrary{patterns}
\begin{figure}
\begin{center}
\begin{tikzpicture}[>=latex,scale=1.4]

\draw[pattern=north east lines,pattern color=gray!50] (3.5,0) circle (0.8);
\draw[pattern=north west lines,pattern color=gray!50] (3.5,0) circle (0.6);
\draw[fill] (3.5,0) circle (1.5pt);
\node[below, inner sep =3.5pt] at (3.5,0) {$x_\ast$};
\draw[dashed] (3.5,0) circle (1.3);

\draw[gray,pattern=north east lines,pattern color=gray!50] (-3,0) circle (1.2);
\draw[pattern=north west lines,pattern color=gray!50](-3,0) circle [x radius=0.5, y radius=1.2];
\draw[fill] (-3,0) circle (1.5pt);
\node[below, inner sep =3.5pt] at (-3,0) {$-\rho_dx_\ast$};
\draw[dashed] (-3,0) circle (1.3);

\draw[fill,gray!30] (0,0) circle (1);
\draw[fill,white] (0,0) circle (0.5);
\draw[fill] (0,0) circle (1.5pt);
\node[below, inner sep =3.5pt] at (0,0) {$0$};

\node at (3.5,0.3) {$\mathcal S_\beta^+$};
\node at (-3,0.3) {$\mathcal S_\beta^-$};

\draw[->] (-0,0.2) -- (-0.4,0.2);
\draw[->] (-0.4,0.2) -- (-0.8,0.2);
\draw[->] (-0.8,0.2) -- (-1.2,0.2);
\draw[->] (-1.2,0.2) -- (-1.6,0.2);
\draw[->] (-1.6,0.2) -- (-2,0.2);
\draw[->] (-2,0.2) .. controls (0.25,1.5) .. (2.5,0.2);
\draw[->] (2.5,0.2) -- (2.9,0.2);
\draw[->] (2.9,0.2) -- (3.3,0.2);
\end{tikzpicture}
\end{center}
\vspace{-0.3cm}
\caption{
\label{sec:ProofLogic}
Logic of the proof, explained in the text.
}
\end{figure}
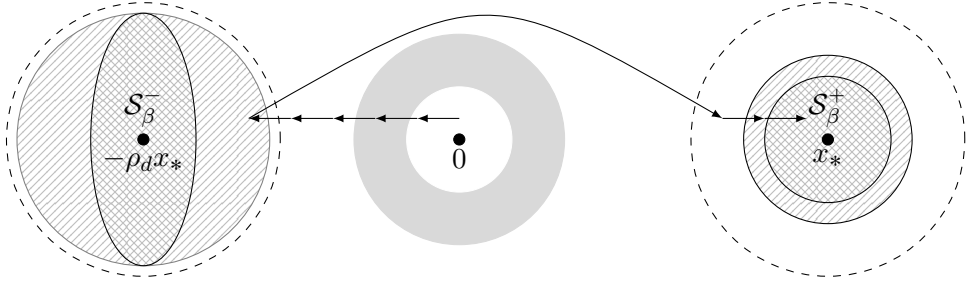


\section{Applications to Compressed Sensing}

In this section we briefly discuss another important scenario to which our results apply to, namely regularizing inverse problems using deep generative priors.
Approaches that regularize inverse problems using deep generative models~\citep{bora_compressed_2017} have empirically been shown to improve over sparsity-based approaches, see~\citep{lucas_using_2018} for a review for applications in imaging, and~\citep{mardani_recurrent_2017} for an application in Magnetic Resonance Imaging showing a significant performance improvement over conventional methods.

Consider an inverse problem, where the goal is to reconstruct an unknown vector $y_\ast \in \R^n$ from $m<n$ noisy linear measurements:
\[
z= A y_\ast + \eta
\quad \in \R^m,
\]
where $A \in \R^{m\times n}$ is called the measurement matrix and $\eta$ is zero mean Gaussian noise with covariance matrix $\sigma^2/n I$, as before.
As before, assume that $y_\ast$ lies in the range of a generative prior $G$, i.e., $y_\ast = G(\xo)$ for some $\xo$.
As a way to recover $\xo$, consider minimizing the empirical risk objective $f(x) = \frac{1}{2} \norm{A G(x) - z}$, using Algorithm~\ref{alg}, with Step~6 substituted by
$\stepdir{x_i} = (A\PiWdi)^t (AG(x_i) - y)$, to account for the fact that measurements were taken with the matrix $A$.

Suppose that $A$ is a random projection matrix, for concreteness assume that $A$ has i.i.d.~Gaussian entries with variance $1/m$.
One could prove an analogous result as Theorem~\ref{thm:main-rescaled}, but with $\omega = \sqrt{ 18\sigma^2 \frac{k}{m}
\log(n_1^d n_2^{d-1}\ldots n_d)}$, (note that $n$ has been replaced by $m$).  This extension shows that, provided $\epsilon$ is chosen sufficiently small, that our algorithm yields an iterate $x_i$ obeying
\[
\norm{G(x_i) - G(\xo)}^2
\lesssim
\sigma^2 \frac{k}{m},
\]
where again $\lesssim$ absorbs factors logarithmic in the $n_i$'s, and polynomial in $d$.  Proving this result would be analogous to the proof of Theorem~\ref{thm:main-rescaled}, but with the additional assumption that the sensing matrix $A$ acts like an isometry on the union of the ranges of $\PiWdi$, analogous to the proof in \citep{hand_global_2017}.
This extension of our result shows that Algorithm~\ref{alg} enables solving inverse problems under noise efficiently, and quantifies the effect of the noise.

We hasten to add that the paper~\citep{bora_compressed_2017} also derived an error bound for minimizing empirical loss.
However, the corresponding result (for example Lemma~4.3) differs in two important aspects to our result.  First, the result in~\citep{bora_compressed_2017} only makes a statement about the \emph{minimizer} of the empirical loss and does not provide justification that an \emph{algorithm} can efficiently find a point near the global minimizer.  As the program is non-convex, and as non-convex optimization is NP-hard in general,  the empirical loss could have local minima at which algorithms get stuck.  In contrast, the present paper presents a specific practical algorithm and proves that it finds a solution near the global optimizer regardless of initialization.
Second, the result in~\citep{bora_compressed_2017} considers arbitrary noise $\noise$ and thus can not assert denoising performance.  
In contrast, we consider a random model for the noise, and show the denoising behavior that the resulting error is no more than $O(k/n)$,
as opposed to $\norm{\noise}^2 \approx O(1)$, which is what we would get from direct application of the result in~\citep{bora_compressed_2017}.

\section{Experimental results}

In this section we provide experimental evidence that corroborates our theoretical claims that denoising with deep priors achieves a denoising rate proportional to $\sigma^2 k/n$.
We focus on denoising by enforcing a generative prior, and consider both a synthetic, random prior, as studied theoretically in the paper, as well as a prior learned from data.
All our results are reproducible with the code provided in the supplement.

\subsection{Denoising with a synthetic prior}

We start with a synthetic generative network prior with ReLU-activation functions, and draw its weights independently from a Gaussian distribution.
We consider a two-layer network with $n=1500$ neurons in the output layer, $500$ in the middle layer, and vary the number of input neurons, $k$, and the noise level, $\sigma$.
We next present simulations showing that if $k$ is sufficiently small, our algorithm achieves a denoising rate proportional to $\sigma k/n$ as guaranteed by our theory.

Towards this goal, we generate Gaussian inputs $\xo$ to the network and observe the noisy image $y = G(\xo) + \noise$, $\noise \sim \mathcal N(0,\sigma^2/n I)$.
From the noisy image, we first obtain an estimate $\hat x$ of the latent representation by running Algorithm~\ref{alg} until convergence, and second we obtain an estimate of the image as $\hat y = G(\hat x)$.
In the left and middle panel of Figure~\ref{fig:plots}, we depict the normalized mean squared error of the latent representation,
$\MSE(\hat x, \xo)$, and the mean squared error in the image domain, $\MSE(G(\hat x), G(\xo))$, where we defined
$\MSE(z, z') = \norm{z - z'}^2$.
For the left panel, we fix the noise variance to $\sigma^2=0.25$, and vary $k$, and for the middle panel we fix $k=50$ and vary the noise variance.
The results show that, if the network is sufficiently expansive, guaranteed by $k$ being sufficiently small, then in the noiseless case ($\sigma^2 = 0$), the latent representation and image are perfectly recovered. In the noisy case, we achieve a MSE proportional to $\sigma^2k/n$, both in the representation and image domains.

We also observed that for the problem instances considered here, the negation trick in step 3-4 of Algorithm~\ref{alg} is often not necessary, in that even without that step the algorithm typically converges to the global minimum.
Having said this, in general the negation step is necessary,  since there exist problem instances that have a local minimum opposite of $\xo$.

\subsection{Denoising with a learned prior}

We next consider a prior learned from data. Technically, for such a prior our theory does not apply since we assume the weights to be chosen at random.
However, the numerical results presented in this section show that even for the learned prior we achieve the rate predicted by our theory pertaining to a random prior.
Towards this goal, we consider a  fully-connected autoencoder parameterized by $k$, consisting of an decoder and encoder with ReLU activation functions and fully connected layers. We choose the number of neurons in the three layers of the encoder as $784,400,k$, and those of the decoder as $k,400,784$.
We set $k = 10$ and $k=20$ to obtain two different autoencoders.
We train both autoencoders on the MNIST~\citep{lecun_gradient-based_1998} \emph{training set}.

We then take an image $y_\ast$ from the MNIST \emph{test set}, add Gaussian noise to it, and denoise it using our method based on the learned decoder-network $G$ for $k=10$ and $k=20$.
Specifically, we estimate the latent representation $\hat x$ by running Algorithm~\ref{alg}, and then set $\hat y = G(\hat x)$. See Figure~\ref{fig:examples} for a few examples demonstrating the performance of our approach for different noise levels.

\begin{figure}
\begin{center}
\begin{tikzpicture}[>=latex,scale=0.8]

\draw[->]  (-7,0.7) node[left] {noisy} -- (-6.6,0.7);
\draw[->]  (-7,-0.5) node[left] {denoised} -- (-6.6,-0.5);
\draw[->]  (-7,1.5) node[left] {noise variance} -- (-6.6,1.5);

\foreach \x/\xtext in {-5.9/0, -4.59/0.3, -3.28/0.6, -1.97/0.9, -0.65/1.2, 0.65/1.5, 1.97/1.8, 3.28/2.1, 4.59/2.4, 5.9/2.7}
\node at (\x,1.5) {\xtext};

\node at (0,0){ \includegraphics[width=10.5cm]{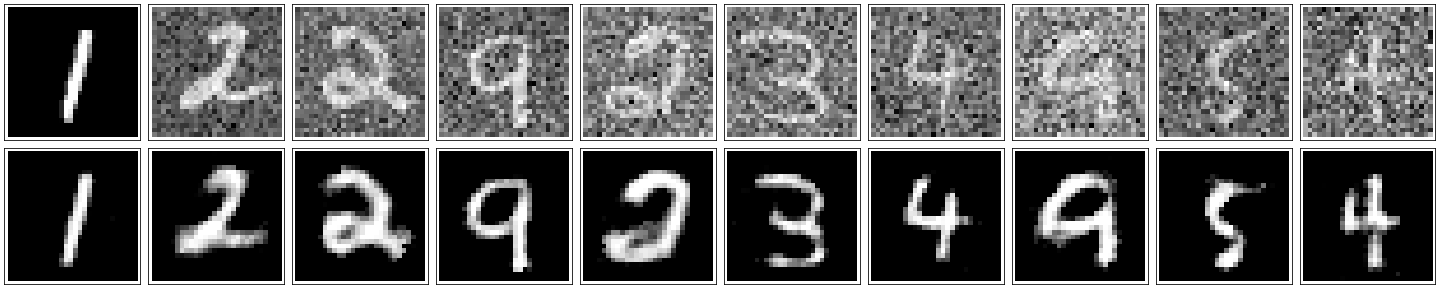}};
\end{tikzpicture}
\end{center}
\vspace{-0.5cm}
\caption{
\label{fig:examples}
Denosing with a learned generative prior:
Even when the number is barely visible, the denoiser recovers a sharp image.
}
\end{figure}
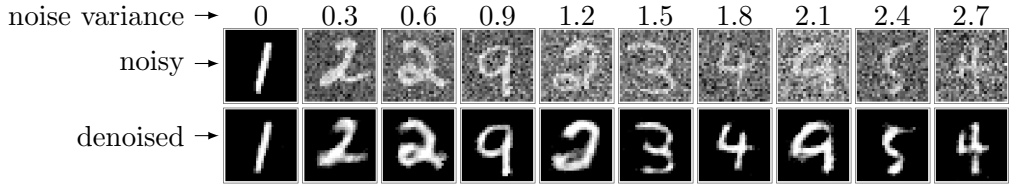

We next show that this achieves a mean squared error (MSE) proportional to $\sigma^2 k/n$, as suggested by our theory which applies for decoders with random weights.
We add noise to the images with noise variance ranging from $\sigma^2 = 0$ to $\sigma^2 = 6$.
In the right panel of Figure~\ref{fig:plots} we show the MSE in the image domain, $\MSE(G(\hat x), G(\xo))$, averaged over a number of images for the learned decoders with $k=10$ and $k=20$.
We observe an interesting tradeoff:
The decoder with $k=10$ has fewer parameters, and thus does not represent the digits as well, therefore the MSE is larger than that for $k=20$ for the noiseless case (i.e., for $\sigma=0$).
On the other hand, the smaller number of parameters results in a better denoising rate (by about a factor of two), corresponding to the steeper slope of the MSE as a function of the noise variance, $\sigma^2$.

\begin{figure}
\begin{center}
\begin{tikzpicture}
\begin{groupplot}[
yticklabel style={
        /pgf/number format/fixed,
        /pgf/number format/precision=5
},
scaled y ticks=false,
         title style={at={(0.5,-0.35)},anchor=north},
         group style={group size=3 by 1, horizontal sep=1cm},
         width=0.33\textwidth, legend pos=north west]

\nextgroupplot[xlabel={$k$}, thick,ylabel={mean squared error}]

\addplot +[mark=noe,red] table[x index=0,y index=2]{./fig/denoise_random_dnn_k.dat};
\addlegendentry{\scriptsize noisy rep.}

\addplot +[mark=noe,blue] table[x index=0,y index=1]{./fig/denoise_random_dnn_k.dat};
\addlegendentry{\scriptsize noisy img.}

\nextgroupplot[xlabel={$\sigma^2$},ymax = 0.6,thick]

\addplot +[mark=noe,red] table[x index=0,y index=2]{./fig/denoise_random_dnn_sigma.dat};
\addlegendentry{\scriptsize noisy rep.}

\addplot +[mark=noe,blue] table[x index=0,y index=1]{./fig/denoise_random_dnn_sigma.dat};
\addlegendentry{\scriptsize noisy img.}

\nextgroupplot[xlabel={$\sigma^2$},ymax = 0.6,thick]

\addplot +[mark=noe,blue] table[x index=0,y index=1]{./fig/denoise_learned.dat};
\addlegendentry{\scriptsize$k=10$}

\addplot +[mark=noe,red] table[x index=0,y index=2]{./fig/denoise_learned.dat};
\addlegendentry{\scriptsize$k=20$}

\addplot +[mark=noe,dotted,blue] table[x index=0,y index=3]{./fig/denoise_learned.dat};

\addplot +[mark=noe,dotted,red] table[x index=0,y index=4]{./fig/denoise_learned.dat};

\end{groupplot}
\end{tikzpicture}
\end{center}
\vspace{-0.5cm}
\caption{\label{fig:plots}
Mean square error in the image domain, $\MSE(G(\hat x), \xo)$, and in the latent representation, $\MSE(\hat x, \xo)$, as a function of the dimension of the latent representation, $k$, with $\sigma^2=0.25$ {\bf(left panel)}, and the noise variance, $\sigma^2$ with $k=50$ {\bf (middle panel)}.
As suggested by the theory pertaining to decoders with random weights, if $k$ is sufficiently small, and thus the network is sufficiently expansive, the denoising rate is proportional to $\sigma^2 k/n$.
{\bf Right panel:}
Denoising of handwritten digits based on a learned decoder with $k=10$ and $k=20$, along with the least-squares fit as dotted lines.
The learned decoder with $k=20$ has more parameters and thus represents the images with a smaller error; therefore the MSE at $\sigma=0$ is smaller.
However, the denoising rate for the decoder with $k=20$, which is the slope of the curve is larger as well, as suggested by our theory.
}
\end{figure}
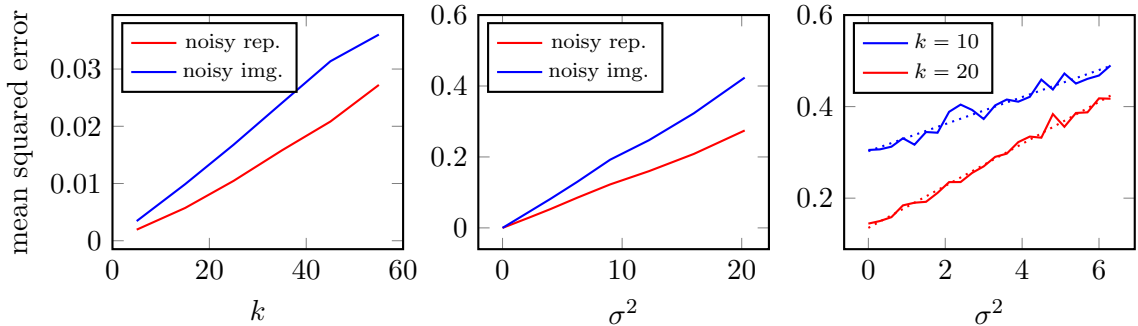

\subsection*{Acknowledgements}

RH is partially supported by a NSF Grant ISS-1816986 and
PH is partially supported by a NSF CAREER Grant DMS-1848087 as well as NSF Grant DMS-1464525, and the authors would like to thank Tan Nguyen for helpful discussions.

\printbibliography

@inproceedings{heckel_deep_2019,
  title={Deep Decoder: Concise Image Representations from
Untrained Non-convolutional Networks},
  author={Heckel, Reinhard and Hand, Paul},
  booktitle={International Conference on Learning Representations},
  year={2019},
}

@article{laurent_adaptive_2000, title={Adaptive estimation of a quadratic functional by model 			 selection}, volume={28}, number={5}, journal={The Annals of Statistics}, author={Laurent, B. and Massart, P.}, year={2000},  pages={1302--1338} }

@inproceedings{ma2018invertibility,
  title={Invertibility of convolutional generative networks from partial measurements},
  author={Ma, Fangchang and Ayaz, Ulas and Karaman, Sertac},
  booktitle={Advances in Neural Information Processing Systems},
  pages={9651--9660},
  year={2018}
}

@article{clason_2018,
	title = {Nonsmooth Analysis and Optimization},
	journal = {arXiv:1708.04180},
	author = {C. Clason},
	year = {2017},
}

@ARTICLE{Karras17,
   author = {{Karras}, T. and {Aila}, T. and {Laine}, S. and {Lehtinen}, J.
	},
    title = "{Progressive growing of GANs for improved quality, stability, and variation}",
  journal = {arXiv: 1710.10196},
     year = 2017,
    month = oct,
}

@ARTICLE{Arora15,
   author = {{Arora}, S. and {Liang}, Y. and {Ma}, T.},
    title = "{Why are deep nets reversible: A simple theory, with implications for training}",
  journal = {arXiv:1511.05653},
     year = 2015,
}

@inproceedings{DIP,
   author = {{Ulyanov}, D. and {Vedaldi}, A. and {Lempitsky}, V.},
    title = {Deep Image Prior},
    booktitle   = {IEEE Conf. Computer Vision and Pattern Recognition},
     year = 2018,
}

@book{boucheron_concentration_2013,
	title = {Concentration {Inequalities}: {A} {Nonasymptotic} {Theory} of {Independence}},
	publisher = {Oxford University Press},
	author = {Lugosi, G. and Massart, P. and Boucheron, S.},
	year = {2013},
}

@article{bora_compressed_2017,
	title = {Compressed sensing using generative models},
	journal = {arXiv:1703.03208},
	author = {Bora, A. and Jalal, A. and Price, E. and Dimakis, A. G.},
	year = {2017},
}

@inproceedings{hand_global_2017,
	title = {Global guarantees for enforcing deep generative priors by empirical risk},
	booktitle = {Conference on Learning Theory},
	note = {arXiv:1705.07576},
	author = {Hand, P. and Voroninski, V.},
	year = {2018},
}

@article{elad_image_2006,
	title = {Image denoising via sparse and redundantd representations over learned dictionaries},
	volume = {15},
	number = {12},
	journal = {IEEE Transactions on Image Processing},
	author = {Elad, M. and Aharon, M.},
	year = {2006},
	pages = {3736--3745},
}

@article{donoho_de-noising_1995,
	title = {De-noising by soft-thresholding},
	volume = {41},
	number = {3},
	journal = {IEEE Transactions on Information Theory},
	author = {Donoho, D. L.},
	year = {1995},
	pages = {613--627},
}

@article{hinton_reducing_2006,
	title = {Reducing the dimensionality of data with neural networks},
	volume = {313},
	number = {5786},
	journal = {Science},
	author = {Hinton, G. E. and Salakhutdinov, R. R.},
	year = {2006},
	pages = {504--507}
}

@article{zhang_beyond_2017,
	title = {Beyond a {Gaussian} denoiser: {Residual} learning of deep {CNN} for image denoising},
	volume = {26},
	number = {7},
	journal = {IEEE Transactions on Image Processing},
	author = {Zhang, K. and Zuo, W. and Chen, Y. and Meng, D. and Zhang, L.},
	year = {2017},
	pages = {3142--3155},
}

@inproceedings{burger_image_2012,
	title = {Image denoising: {Can} plain neural networks compete with {BM}3D?},
	shorttitle = {Image denoising},
	booktitle = {2012 {IEEE} {Conference} on {Computer} {Vision} and {Pattern} {Recognition}},
	author = {Burger, H. C. and Schuler, C. J. and Harmeling, S.},
	year = {2012},
	pages = {2392--2399},
}

@article{dabov_image_2007,
	title = {Image denoising by sparse 3-{D} transform-domain collaborative filtering},
	volume = {16},
	number = {8},
	journal = {IEEE Transactions on Image Processing},
	author = {Dabov, K. and Foi, A. and Katkovnik, V. and Egiazarian, K.},
	year = {2007},
	pages = {2080--2095},
}

@article{lucas_using_2018,
	title = {Using deep neural networks for inverse problems in imaging: {Beyond} analytical methods},
	volume = {35},
	number = {1},
	journal = {IEEE Signal Processing Magazine},
	author = {Lucas, A. and Iliadis, M. and Molina, R. and Katsaggelos, A. K.},
	year = {2018},
	pages = {20--36},
}

@article{mardani_recurrent_2017,
	title = {Recurrent generative adversarial networks for proximal learning and automated compressive image recovery},
	journal = {arXiv:1711.10046},
	author = {Mardani, M. and Monajemi, H. and Papyan, V. and Vasanawala, S. and Donoho, D. and Pauly, J.},
	year = {2017},
}

@article{starck_curvelet_2002,
	title = {The curvelet transform for image denoising},
	volume = {11},
	number = {6},
	journal = {IEEE Transactions on Image Processing},
	author = {Starck, Jean-Luc and Candes, E. J. and Donoho, D. L.},
	year = {2002},
	pages = {670--684},
}

@incollection{goodfellow_generative_2014,
	title = {Generative adversarial nets},
	booktitle = {Advances in {Neural} {Information} {Processing} {Systems} 27},
	author = {Goodfellow, I. and Pouget-Abadie, J. and Mirza, M. and Xu, B. and Warde-Farley, D. and Ozair, S. and Courville, A and Bengio, Y.},
	year = {2014},
	pages = {2672--2680},
}

@article{lecun_gradient-based_1998,
	title = {Gradient-based learning applied to document recognition},
	volume = {86},
	number = {11},
	journal = {Proceedings of the IEEE},
	author = {Lecun, Y. and Bottou, L. and Bengio, Y. and Haffner, P.},
	year = {1998},
	pages = {2278--2324},
}


\appendix

\section{Proof of Proposition~\ref{prop:autodec}}

We first show that $H(y)$ lies in the range of a union of $k$-dimensional subspaces, and upper-bound the number of the subspaces. Towards this goal, first note that the effect of the ReLU operation $\relu(z)$ can be described with a diagonal matrix $D$ which contains a one on its diagonal if the respective entry of $z$ is larger than zero, and zero otherwise, so that $D z = \relu(z)$. With this notation, we can write
\[
H(y) 
= 
\underbrace{D_d W_d D_{d-1} \ldots D_2 W_2 D_1 W_1 D' W'}_{U} y.
\]
The matrix $U \in \reals^{n\times n}$ has at most rank $k$, thus $H(y)$ lies in the range of a union of at most $k$-dimensional subspaces, where each subspace is determined by the matrices $D_d,\ldots, D_1,D'$.
We next bound the number of subspaces.
First note that since $W' \in \reals^{n\times k}$, there are only $2^k$ many different choices for $D'$, corresponding to all the sign patterns.
Next note that, with the lemma below, we have that for a fixed $D' W'$, the number of different matrixes $D_1$ can be bounded by $n_1^k$. Likewise, for fixed 
$W_2 D_1 W_1 D' W'$, the number of different matrices $D_2$ can be bounded by $n_2^k$ and so forth. 
Thus, the total number of different choices of the matrices $D_d,\ldots, D_1,D'$ is upper bounded by
\[
2^k n_1^k \ldots n_d^k.
\]

\begin{lemma}
\label{lem:signpatterns}
For any $U \in \reals^{n\times k}$ and $k\geq 5$,
\[
| \{ \diag(U v > 0) U | v \in \reals^k \} |
\leq 
n^{k}.
\]
\end{lemma}


Next note that by assumption we have that 
\begin{align}
\label{eq:ineqU}
\norm[2]{U \eta }^2 / \norm[2]{\eta}^2 \leq 2,
\end{align}
for all vectors $\eta$ and for all $U$ defined by the matrices $D_d,\ldots, D_1,D'$. 
For fixed $U$, let $S$ be the span of the right singular vectors of $U$, and note that $S$ has dimension at most $k$.  Let $\PS$ be the orthogonal projector onto a subspace $S$.  We have that 
\begin{align*}
\frac{\norm[2]{U \eta}^2}{ \norm[2]{\eta}^2} =\frac{\norm[2]{U P_S \eta}^2}{ \norm[2]{\eta}^2}   \leq 2\frac{\norm[2]{\PS \eta}^2}{\norm[2]{ \eta}^2},
\end{align*}
again for all $\eta$.
Now, we make use of the following bound on the projection of the noise $\eta$ onto a subspace, which follows from standard Gaussian concentration inequalities~\citep[Lem.~1]{laurent_adaptive_2000}. 
\begin{lemma}
\label{lem:noise-projection}
Let $S \subset \reals^n$ be a subspace with dimension $k$.  Let $\eta \sim \mathcal{N}(0, I_n)$ and $\beta \geq 1$.  Then,
\[
\PR{ \frac{\norm[2]{\PS \eta}^2}{\norm[2]{ \eta}^2} \leq \frac{10 \beta k}{n}} \geq 1 - e^{-\beta k} - e^{-n/16}.
\]
\end{lemma}

Taking a union bound over all subspaces, we obtain with the lemma above that
\[
\PR{ \frac{\norm[2]{H(\eta)}^2}{\norm[2]{ \eta}^2}
\leq 
\frac{20 \beta k}{n}x
}
\geq
1 - (2^k n_1^k \ldots n_d^k) (e^{-\beta k} + e^{-n/16}).
\]
Choosing $\beta = 2 \log(2n_1 n_2 \ldots n_d)$ concludes the proof.

\section{Proof of Theorem~\ref{thm:main-rescaled}} \label{sec:proofs}

In this section we prove our main result, Theorem~\ref{thm:main-rescaled}.   Instead of proving Theorem~\ref{thm:main-rescaled} as stated, we will prove the following equivalent rescaled statement for when $W_i$ have i.i.d. $\mathcal{N}(0, 1/{n_i})$ entries.  Because of this rescaling, $G(x)$ scales like $2^{-d/2} \norm{x}$, the noise $\omega$ is assumed to scale like $2^{-d/2}$, $\nabla f$ scales like $2^d$, and $\alpha$ scales like $2^d$.  Theorem~\ref{thm:main-rescaled} is the $\epsilon=K/d^{90}$ case of what follows.  

\begin{theorem}
\label{thm:main}
Consider a network with the weights in the $i$-th layer, $W_i \in \R^{n_i \times n_{i-1}}$, i.i.d.~$\mathcal N(0,1/n_i)$ distributed, and suppose that the network satisfies the expansivity condition for some $\epsilon \leq K/d^{90}$. 
Also, suppose that the noise variance obeys
\[
\omega \leq \frac{\norm{\xo} K_1 2^{-d/2} }{d^{16}},
\quad
\omega
\defeq
\sqrt{ 18\sigma^2 \frac{k}{n}
\log(n_1^d n_2^{d-1}\ldots n_d)}.
\]
Consider the iterates of Algorithm~\ref{alg} with stepsize
$\alpha = K_4\frac{2^d}{d^2}$.
\begin{enumerate}
\item[A.]
Then, there exists a number of steps $N$ upper bounded by
\[
N
\leq
\frac{K_2}{d^4\eps}
\frac{f(x_0) 2^d}{ \norm{\xo}}
\]
such that after $N$ steps, the iterates of Algorithm~\ref{alg} obey
\begin{align}
\label{eq:errorxixo}
\norm{x_i -\xo}
\leq
K_5 d^9 \norm{\xo} \sqrt{\epsilon}
+
K_6 d^6 2^{d/2} \omega,
\quad
\text{for all $i \geq N$},
\end{align}
with probability at least
$1-2e^{-2k \log n} - \sum_{i=2}^d 8 n_i e^{-K_7 n_{i-2} }
- 8 n_1 e^{-K_7 \epsilon^{2} \log(1/\epsilon) k }
$.
\item[B.]
In addition, for all $i \geq N$, we have
\begin{align}
\|x_{i + 1} - x_*\| &\leq (1 - \alpha 7/8)^{i + 1 - N} \|x_N - x_*\| + K_8 2^{d/2} \omega \hbox{ and }  \label{GA:e68:2} \\
\|G(x_{i + 1}) - G(x_*)\| &\leq \frac{1.2}{2^{d/2}} (1 - \alpha 7/8)^{i + 1 - N} \|x_N - x_*\| + 1.2 K_8 \omega, \label{GA:e74:2}
\end{align}
where $\alpha < 1$ is the stepsize of the algorithm.
Here, $K_1,K_2,..$ are numerical constants, and $x_0$ is the initial point in the optimization.
\end{enumerate}
\end{theorem}

As mentioned in Section~\ref{sec:WDC}, our proof makes use of a deterministic condition, called the Weight Distribution Condition (WDC), formally defined in Section~\ref{sec:WDC}.
The following proposition establishes that the expansivity condition ensures that the WDC holds:

\begin{lemma}[Lemma 9 in~\citep{hand_global_2017}] \label{thm-multi-layer}
Fix $\epsilon \in (0,1)$. If the entires of $W_i \in \R^{n_i \times n_{i-1}}$ are i.i.d. $\mathcal{N}(0, 1/n_i)$ and the expansivity condition 
\[
n_i > c \eps^{-2} \log(1/\eps) n_{i-1} \log n_{i-1}
\]
holds,
then $W_i$ satisfies the WDC with constant $\eps$ with probability at least $1 - 8 n_i e^{- K \epsilon^2 n_{i-1} }$.
Here, $c$ and $K$ are numerical constants.
\end{lemma}

It follows from Lemma~\ref{thm-multi-layer}, that the WDC holds for all $W_i$ simultaneously with probability at least $1- \sum_{i=2}^d 8 n_i e^{-K_7 n_{i-2} }
- 8 n_1 e^{-K_7 \epsilon^{2} \log(1/\epsilon) k }$.

In the remainder of the proof we work on the event that the WDC holds for all $W_i$.

\subsection{Preliminaries}

%

Recall that the goal of our algorithm is to minimize the empirical risk objective
\begin{align*}
f(x) = \frac{1}{2} \norm{ G(x) - y }^2, 
\end{align*}
where $y \defeq G(\xo) + \eta$, with $\eta \sim \mathcal N(0,\sigma^2/n I)$.

Our results rely on the fact that outside of two balls around $x = \xo$ and $x = - \rho_d \xo$, with $\rho_d$ a constant defined below, the direction chosen by the algorithm is a descent direction, with high probability.
In order to prove this, we use a concentration argument, similar to the arguments used in~\citep{hand_global_2017}.
First, define 
\[
\Lambda_x \defeq \PiWd,
\]
with $W_{i,+,x}$ defined in Section~\ref{sec:erm} for notational convenience, and note that the step direction of our algorithm can be written as
\begin{align}
\stepdir{x} = \vbarx + \qbarx,
\quad
\text{with}
\quad
\vbarx :=
 \Lambda_x^t  \Lambda_x x -  (\Lambda_x)^t (\Lambda_{\xo}) \xo,
 \quad
 \text{and}
 \quad
 \qbarx := \Lambda_x^t \eta.
\label{defn-vxxo}
\end{align}
Note that at points $x$ where $G$ (and hence $f$) is differentiable, we have that $\stepdir{x} = \nabla f(x)$.

The proof is based on showing that $\stepdir{x}$ concentrates around a particular $\hx \in \R^k$, defined below, that is a continuous function of nonzero $x,\xo$ and is zero only at $x = \xo$ and $x = - \rho_d \xo$.
The definition of $\hx$ depends on a function that is helpful for controlling how the operator $x \mapsto W_{+,x} x$ distorts angles, defined as:\begin{align}
g(\theta) \defeq \cos^{-1} \Bigl( \frac{ (\pi - \theta) \cos \theta + \sin \theta}{\pi} \Bigr). \label{defn-g}
\end{align}
With this notation, we define
\begin{align*}
\hx \defeq - \frac{1}{2^d} \Bigl( \prod_{i=0}^{d-1} \frac{\pi - \thetabar_i}{\pi}  \Bigr)\xo
+ \frac{1}{2^d} \left[ x - \sum_{i=0}^{d-1} \frac{\sin \thetabar_i}{\pi}  \Bigl( \prod_{j=i+1}^{d-1} \frac{\pi - \thetabar_j}{\pi}  \Bigr)  \frac{\|\xo\|_2}{\|x\|_2} x  \right] ,
\end{align*}
where $\thetabar_0 = \angle(x, \xo)$ and $\thetabar_i = g(\thetabar_{i-1})$.
Note that $h_x$ is deterministic and only depends on $x$, $\xo$, and the number of layers, $d$.  

In order to bound the deviation of $\stepdir{x}$ from $h_x$ we use the following two lemmas, bounding the deviation controlled by the WDC and the deviation from the noise:

\begin{lemma}[Lemma~6 in~\citep{hand_global_2017}]
\label{lemma:vxbarhx}
Suppose that the WDC holds with $\eps < 1/(16 \pi d^2)^2$. Then, for all nonzero $x, \xo \in \R^k$,
\begin{align}
\| \vbarx - \hx \|_2 &\leq K \frac{d^3 \sqrt{\eps}}{2^d} \max ( \|x\|_2, \| \xo\|_2), \text{ and } \label{vxbarhx-are-close} \\
\bigl \langle  \Lambda_x x, \Lambda_{\xo} \xo \bigr \rangle  &\geq \frac{1}{4 \pi} \frac{1}{2^d} \|x\|_2 \|\xo\|_2, \text{ and } \label{piwipxxpiwipyy} \\
\norm{ \Lambda_x }^2 \leq \frac{1}{2^d} (1 + 2 \eps)^d
&\leq
\frac{13}{12} 2^{-d}.
\label{eq:boundGammax}
\end{align}
\end{lemma}
\begin{proof}
Equation~\eqref{vxbarhx-are-close} and \eqref{piwipxxpiwipyy} are Lemma~6 in~\citep{hand_global_2017}.
Regarding~\eqref{eq:boundGammax}, note that the WDC implies that $\norm{W_{i,+,x}}^2 \leq 1/2 + \epsilon$.
It follows that
\[
\norm{\Lambda_x}^2
=
\norm{\PiWd}^2 \leq \frac{1}{2^d} (1 + 2 \eps)^d
= \frac{1}{2^d} e^{d \log(1+2 \eps)} \leq \frac{1 + 4 \eps d}{2^d}
\leq
\frac{13}{12} 2^{-d},
\]
where the last inequalities follow by our assumption on $\epsilon$ (i.e., $\eps < 1/(16 \pi d^2)^2$).
\end{proof}

\begin{lemma}
\label{lem:noisesmall}
Suppose the WDC holds with $\eps < 1/(16 \pi d^2)^2$, that any subset of $n_{i-1}$ rows of $W_i$ are linearly independent for each $i$,  and that $\eta \sim \mathcal N(0, \sigma^2/ n I)$.
Then the event
\begin{align}
\label{eq:qnoiseb}
\eventnoisesmall
\defeq
\left\{
\norm{ \Lambda_x^t \eta }
\leq
\frac{\omega}{2^{d/2}},
\; \text{ for all $x$}
\right\},
\quad
\omega
\defeq
\sqrt{ 16 \sigma \frac{k}{n}
\log(n_1^d n_2^{d-1}\ldots n_d)}
\end{align}
holds with probability at least $1-2e^{-2 k \log n}$.
\end{lemma}

As the cost function $f$ is not differentiable everywhere, we make use of the generalized subdifferential in order to reference the subgradients at nondifferentiable points. For a Lipschitz function $\tilde{f}$ defined from a Hilbert space $\mathcal{X}$ to $\mathbb{R}$, the Clarke generalized directional derivative of $\tilde{f}$ at the point $x \in \mathcal{X}$ in the direction $u$, denoted by $\tilde{f}^o(x; u)$, is defined by $\tilde{f}^o(x; u) = \limsup_{y \rightarrow x, t \downarrow 0} \frac{\tilde{f}(y + t u) - \tilde{f}(y)}{t}$, and the generalized subdifferential of $\tilde{f}$ at $x$, denoted by $\partial \tilde{f}(x)$, is defined by
\begin{equation*}
\partial \tilde{f}(x) = \{v \in \mathbb{R}^k \mid \inner[]{v}{u} \leq \tilde{f}^o(x; u), \text{ for all } u \in \mathcal{X}\}.
\end{equation*}
Since $f(x)$ is a piecewise quadratic function, we have
\begin{equation} \label{subdifferential-convex-hull}
\partial f(x) = \conv(v_1, v_2, \ldots, v_t),
\end{equation}
where $\conv$ denotes the convex hull of the vectors $v_1, \ldots, v_t$, $t$ is the number of quadratic functions adjoint to $x$, and $v_i$ is the gradient of the $i$-th quadratic function at $x$. 

\begin{lemma} \label{lemma:vxtilde-hxtilde-subdifferential}
Under the assumption of Lemma~\ref{lem:noisesmall}, and assuming that $\eventnoisesmall$ holds, we have that, for any $x\neq 0$ and any  $v_x \in \partial f(x)$,
\[
\norm{v_x - h_x} \leq K \frac{d^3 \sqrt{\eps}}{2^d} \max ( \|x\|_2, \| \xo\|_2) + \frac{\omega}{2^{d/2}}.
\]
\end{lemma}
\begin{proof}
By \eqref{subdifferential-convex-hull}, $\partial f(x) = \conv(v_1, \ldots v_t)$ for some finite $t$,
and thus $v_x = a_1 v_1 + \ldots a_t v_t$ for some $a_1,\ldots,a_t  \geq 0$, $\sum_i a_i = 1$.
For each $v_i$, there exists a $w$ such that $v_i = \lim_{t \downarrow 0} \tilde{v}_{x + t w}$.  On the event $\eventnoisesmall$, we have that for any $x \neq 0$, for any  $\stepdir{x} \in \partial f(x)$
\begin{align}
\norm{\stepdir{x} - \hx}
&=
\norm{\vbarx + \qbarx-\hx} \nonumber \\
&\leq
\norm{\vbarx-\hx} + \norm{\qbarx} \nonumber \\
&\leq
K \frac{d^3 \sqrt{\eps}}{2^d} \max ( \|x\|_2, \| \xo\|_2)
+
\frac{\omega}{2^{d/2}}, \nonumber
\end{align}
where the last inequality follows from Lemmas~\ref{lemma:vxbarhx} and \ref{lem:noisesmall} above.  The proof is concluded by appealing to the continuity of $h_x$ with respect to nonzero $x$, and by  noting that
\[
\norm{v_x - h_x}
\leq \sum_i a_i \norm{v_i - h_x}
\leq K \frac{d^3 \sqrt{\eps}}{2^d} \max ( \|x\|_2, \| \xo\|_2)
+
\frac{\omega}{2^{d/2}},
\]
where we used the inequality above and that $\sum_i a_i = 1$.
\end{proof}

We will also need an upper bound on the norm of the step direction of our algorithm:

\begin{lemma}
Suppose that the WDC holds with $\epsilon < 1/(16 \pi d^2)^2$ and that the event $\eventnoisesmall$ holds with $\omega \leq \frac{2^{-d/2} \norm{\xo}}{8\pi}$. Then, for all $x$, 
and all $v_x \in \partial f(x)$,
\begin{align}
\label{GA:e23}
\|v_x\|
\leq
\frac{dK}{2^d} \max(\norm{x}, \norm{\xo}),
\end{align}
where $K$ is a numerical constant.
\end{lemma}
\begin{proof}
Define for convenience $\zeta_j=\prod_{i = j}^{d - 1} \frac{\pi - \bar{\theta}_{j, x, x_*}}{\pi}$.
We have
\begin{align}
\|v_x\|
\leq&
\|h_{x}\| + \|h_{x} - v_x\| \nonumber \\
\leq&
\left\|\frac{1}{2^d} x - \frac{1}{2^d} \zeta_{0} x_* - \frac{1}{2^d} \sum_{i = 0}^{d - 1} \frac{\sin \bar{\theta}_{i,x}}{\pi} \zeta_{i + 1}  \frac{\norm{\xo}}{\norm{x}} x \right\|
+
K_1 \frac{d^3 \sqrt{\eps}}{2^d} \max ( \|x\|_2, \| \xo\|_2)
+
\frac{\omega}{2^{d/2}}
\nonumber \\
\leq&
\frac{1}{2^d} \|x\| + \left( \frac{1}{2^d}  + \frac{d}{\pi 2^d} \right) \|x_*\| + K_1 \frac{d^3 \sqrt{\epsilon}}{2^d} \max(\|x\|, \|x_*\|) + \frac{\omega}{2^{d/2}} \nonumber \\
\leq&
\frac{dK}{2^d} \max(\norm{x}, \norm{\xo}), \nonumber
\end{align}
where the second inequality follows from the definition of $h_x$ and Lemma~\ref{lemma:vxtilde-hxtilde-subdifferential}, the third inequality uses $| \zeta_j | \leq 1$,
and the last inequality uses the assumption $\omega \leq \frac{2^{-d/2} \norm{\xo}}{8\pi}$.
\end{proof}



\subsection{Proof of Theorem~\ref{thm:main}A}

We are now ready to prove Theorem~\ref{thm:main}A.
The logic of the proof is illustrated in Figure~\ref{sec:ProofLogic}.
Recall that $x_i$ is the $i$th iterate of $x$ as per Algorithm \ref{alg}.
We first ensure that we can assume throughout that $x_i$ is bounded away from zero:

\begin{lemma}
\label{lem:notinzeroball}
Suppose that WDC holds with $\epsilon < 1/ (16 \pi d^2)^2$ and that $\eventnoisesmall$ holds with $\omega$ in \eqref{eq:qnoiseb} obeying $\omega \leq \frac{2^{-d/2} \norm{\xo}}{8\pi}$.
Moreover, suppose that the step size in Algorithm~\ref{alg} satisfies $0<\alpha < \frac{K 2^d}{d^2}$, where $K$ is a numerical constant.
Then, after at most $N= (\frac{38 \pi K_0 2^d}{\alpha})^2$ steps, we have that for all $i>N$ 
\whcomm{}{ and all $t \in [0, 1]$ that $t \tilde{x}_i + (1 - t) x_{i + 1} \notin \mathcal{B}(0, \frac{1}{32\pi} \|x_*\|)$.}
\end{lemma}

In particular, if  $\alpha = K2^d/d^ 2$, then $N$ is bounded by a constant times $d^4$.

We can therefore assume throughout this proof that $x_i \notin \mathcal B(0, K_0 \norm{\xo})$, $K_0=\frac{1}{32\pi}$.
We prove Theorem~\ref{thm:main} by showing that if $\norm{h_x}$ is sufficiently large, i.e., if the iterate $x_i$ is outside of set
\begin{align*}
\setS_{\beta} &= \Bigl \{ x \in \R^k \mid \| h_x \| \leq \frac{1}{2^d} \beta \max(\|x\|, \|\xo\| ) \Bigr \},
\end{align*}
with
\begin{align}
\label{eq:choicebeta}
\beta = 4 K d^3 \sqrt{\epsilon}
+
13 \omega 2^{d/2}/\norm{\xo},
\end{align}
then the algorithm makes progress in the sense that
$f(x_{i+1}) - f(x_i)$ is smaller than a certain negative value.
The set $\setS_\beta$ is contained in two balls around
$\xo$ and $-\rho \xo$, whose radius is controlled by $\beta$:

\begin{lemma} \label{lemma:Sepsst}
For any $\beta \leq \frac{1}{64^2 d^{12}}$,
\begin{align}
\setS_\beta
\subset
\mathcal{B}(\xo, 5000 d^6 \beta \|\xo\|_2  ) \cup  \mathcal{B}( -\zetacheck_d \xo, 500 d^{11} \sqrt{\beta} \|\xo\|_2).
 \label{Seps-estimate-balls}
\end{align}
Here, $\rho_d>0$ is defined in the proof and obeys $\rho_d \to 1$ as $d \to \infty$.
\end{lemma}

Note that by the assumption $\omega \leq \frac{\norm{\xo} K_1 2^{-d/2} }{d^{16}}$ and $K d^{45} \sqrt{\eps} \leq 1$,
our choice of $\beta$ in \eqref{eq:choicebeta} obeys $\beta \leq \frac{1}{64^2 d^{12}}$ for sufficiently small $K_1, K$, and thus Lemma~\ref{lemma:Sepsst} yields:
\[
\setS_\beta
\subset
\mathcal{B}(\xo, r)
\cup
\mathcal{B}( -\zetacheck_d \xo, \sqrt{r \norm{\xo}} d^8).
\]
were we define the radius
$
r = K_2 d^9 \sqrt{\epsilon} \norm{\xo} + K_3 d^6 \omega 2^{d/2}
$, where $K_2,K_3$ are numerical constants.
Note that hat the radius $r$ is equal to the right hand side in the error bound~\eqref{eq:errorxixo} in our theorem.
In order to guarantee that the algorithm converges to a ball around $\xo$, and not to that around $-\rho_d \xo$, we use the following lemma:

\begin{lemma} \label{GA:le14}
Suppose that the WDC holds with $\eps < 1 / (16 \pi d^2)^2$.
Moreover suppose that $\eventnoisesmall$ holds, and that $\omega$ in the event $\eventnoisesmall$ obeys
$\frac{ \omega }{2^{-d/2} \|\xo\|_2} \leq K_9 / d^2$, where $K_9<1$ is a universal constant. Then for any $\phi_d \in [\rho_d, 1]$, it holds that
\begin{align} \label{GA:e25}
f(x) < f(y)
\end{align}
for all $x \in \mathcal{B}(\phi_d x_*, K_3 d^{-10} \|x_*\|)$ and $y \in \mathcal{B}(- \phi_d x_*, K_3 d^{-10} \|x_*\|)$, where $K_3 < 1$ is a universal constant.
\end{lemma}

In order to apply Lemma~\ref{GA:le14}, define for convenience the two sets:
\begin{align*}
\setS_\beta^+ \defeq&  \setS_\beta \cap \mathcal{B}(x_*, r), \hbox{ and }\\
\setS_\beta^- \defeq&  \setS_\beta \cap \mathcal{B}(- \rho_d x_*, \sqrt{r \norm{\xo}} d^8 ).
\end{align*}

By the assumption that $K d^{45} \sqrt{\eps} \leq 1$ and $\omega \leq K_1 d^{-16} 2^{-d/2} \norm{\xo}$, we have that for sufficiently small $K_1, K$,
\[
\setS_\beta^+
\subseteq
\mathcal{B}(x_*, K_3 d^{-10} \|x_*\|)
\quad
\text{and}
\quad
\setS_\beta^-
\subseteq
\mathcal{B}(- \rho_d x_*, K_3 d^{-10} \|x_*\|).
\]
Thus, the assumptions of Lemma~\ref{GA:le14} are met, and the lemma implies that for any $x \in \setS_{\beta}^-$ and $y \in \setS_{\beta}^+$, it holds that $f(x) > f(y)$.
 We now show that the algorithm converges to a point in $\setS_\beta^+$.  This fact and the negation step in our algorithm (line 3-5) establish that the algorithm converges to a point in $\setS_\beta^+$ if we prove that the objective is nonincreasing with iteration number, which will form the remainder of this proof.

Consider $i$ such that $x_i \notin \setS_\beta$.  
 By the mean value theorem \cite[Theorem~8.13]{clason_2018}, there is a $t \in [0,1]$ such that for $\hat{x}_{i} = x_{i} - t \alpha \stepdir{x_i}$ there is a $v_{\hat{x}_{i}} \in \partial f(\hat{x}_{i})$, where $\partial f$ is the generalized subdifferential of $f$, obeying
\begin{align}
f(x_{i} - \alpha \stepdir{x_i}) - f(x_{i})
=&
\inner[]{v_{\hat{x}_{i}}}{- \alpha \stepdir{x_i}} \nonumber \\
=&
\inner[]{\stepdir{x_i}}{- \alpha \stepdir{x_i}} + \inner[]{v_{\hat{x}_{i}} - \stepdir{x_i}}{- \alpha \stepdir{x_i}} \nonumber \\
\leq&
- \alpha \|\stepdir{x_i}\|^2
+ \alpha \|v_{\hat{x}_{i}} - \stepdir{x_i}\| \|\stepdir{x_i}\| \nonumber \\
=&
- \alpha \|\stepdir{x_i}\|
(\|\stepdir{x_i}\| -  \|v_{\hat{x}_{i}} - \stepdir{x_i}\| ). \label{GA:e21}
\end{align}
In the next subsection, we guarantee that for any $t\in [0,1]$, $v_{\hat{x}_{i}}$ with
$\hat{x}_{i} = x_{i} - t \alpha \stepdir{x_i}$
is close to $\stepdir{x_i}$:
\begin{align}
\|v_{\hat{x}_{i}} - \stepdir{x_i}\|
\leq&
\biggl(\frac{5}{6}  + \alpha K_7 \frac{d^2}{2^d}\biggr)
\norm{ \stepdir{x_i} },  \text { for all } v_{\hat{x}_{i}} \in \partial f(\hat{x}_i).
\label{GA:e24}
\end{align}
Applying~\eqref{GA:e24} to~\eqref{GA:e21} yields
\begin{equation*}
f(x_{i} - \alpha \stepdir{x_i}) - f(x_{i})
\leq
- \frac{1}{12} \alpha \|\stepdir{x_i}\|_2^2,
\end{equation*}
where we used that $\alpha K_7 \frac{d^2}{2^d} \leq \frac{1}{12}$, by our assumption on the stepsize $\alpha$ being sufficiently small.

Thus, the maximum number of iterations for which $x_i \not\in \setS_\beta$ is $f(x_0) 12 / (\alpha \min_i  \norm{\stepdir{x_i}}^2)$.
We next lower-bound $\norm{\stepdir{x_i}}$.
We have that on $\eventnoisesmall$, for all $x \not\in \setS_\beta$, with $\beta$ given by \eqref{eq:choicebeta} that
\begin{align}
\|\stepdir{x} \|_2&\geq \|h_x\| - \| h_x - \stepdir{x}\| \notag \\
&\geq  2^{-d} \max(\|x\|, \|x_*\|) \Bigl( \beta - K_1 d^3 \sqrt{\eps} - \omega\frac{2^{d/2}}{\norm{x_*}}   \Bigr)   \notag\\
&\geq 2^{-d} \max(\|x\|, \|x_*\|) \biggl(3K d^3 \sqrt{\eps} + 12 \omega\frac{2^{d/2}}{\norm{x_*}} \biggr) \label{lowerbound-vxtilde-with-noise} \\
&\geq 2^{-d} \|x_*\| 3K d^3 \sqrt{\eps}, \notag
\end{align}
where the second inequality follows by the definition of $\setS_\beta$ and Lemma \ref{lemma:vxtilde-hxtilde-subdifferential}, and the third inequality follows from our definition of $\beta$ in equation~\eqref{eq:choicebeta}.
Thus,
\begin{align*}
f(x_{i} - \alpha \stepdir{x_i}) - f(x_{i})
\leq
- \alpha K_5 2^{-2d}d^6 \epsilon \norm{\xo}^2
\leq
-2^{-d} d^4 K_6 \epsilon \norm{\xo}^2
\end{align*}
where we used
$\alpha = K_4\frac{2^d}{d^2}$.
Hence, there can be at most $\frac{f(x_0) 2^d}{K_6 d^4 \epsilon \norm{\xo}^2}$ iterations for which $x_i \not \in \setS_\beta$.

In order to conclude our proof, we remark that once $x_i$ is inside a ball of radius $r$ around $\xo$, the iterates do not leave a ball of radius $2r$ around $\xo$.
To see this, note that by the bound on $\|v_x\|$ given in equation~\eqref{GA:e23} and our choice of stepsize,
\[
\alpha \norm{\stepdir{x_i}}
\leq
\frac{K}{d} \max(\norm{x_i}, \norm{\xo}).
\]
This concludes the proof of Theorem~\ref{thm:main}A. 


\subsection{Proof of Theorem~\ref{thm:main}B}

Theorem~\ref{thm:main}A establishes that after $N$ iterations the iterates $x_i$ are inside a ball of radius $2r$ around $\xo$. With the assumption that $\epsilon \leq K_1 / d^{90}$ for sufficiently small $K_1$ and the definition of $r$, this implies that the iterates lie in a ball around $\xo$ of radius at most $K_3 d^{-10} \|x_*\|$. 
In this proof of Theorem~\ref{thm:main}B, we prove convergence within this ball.

In this proof, we show that for any $i \geq N$, it holds that $x_i \in \mathcal{B}(x_*, a_4 d^{-10} \|x_*\|)$, $\tilde{x}_i = x_i$, and
\begin{align*}
\|x_{i + 1} - x_*\| \leq b_2^{i + 1 - N} \|x_N - x_*\| + b_4 2^{d/2} \omega.
\end{align*}
where $K_3$ is defined in Lemma~\ref{GA:le14}, $b_2 = 1 - \frac{\alpha}{2^d} \frac{7}{8}$ and $b_4$ is a universal constant.

We need Lemma~\ref{GA:le10} which guarantees that the search directions of the iterates afterward point to $x_*$ only up to the noise $\omega$:
\begin{lemma} \label{GA:le10}
Suppose the WDC
holds with $200 d \sqrt{d \sqrt{\epsilon}} < 1$ and $x \in \mathcal{B}(x_*, d \sqrt{\epsilon} \|x_*\|)$.  Then
for all $x \neq 0$ and for all $v_x \in \partial f(x)$,
\begin{equation*}
\left\|v_x - \frac{1}{2^d} (x - x_*)\right\| \leq \frac{1}{2^d} \frac{1}{8} \|x - x_*\| + \frac{1}{2^{d/2}} \omega.
\end{equation*}
\end{lemma}

Suppose $\tilde{x}_i \in \mathcal{B}(x_*, K_3 d^{-10} \|x_*\|)$. By the assumption $\epsilon \leq K_1 / d^{90}$ for sufficiently small $K_1$, the assumptions in Lemma~\ref{GA:le10} are met. Therefore,
\begin{align}
\|x_{i + 1} - x_*\|
=& \|\tilde{x}_i - \alpha {v}_{\tilde{x}_i} - x_*\|   \nonumber \\
=& \|\tilde{x}_i - x_*  -  \frac{\alpha}{2^d} (\tilde{x}_i - x_*) - \alpha {v}_{\tilde{x}_i} + \frac{\alpha}{2^d} (\tilde{x}_i - x_*) \| \nonumber \\
\leq& \left( 1 - \frac{\alpha}{2^d} \right) \|\tilde{x}_i - x_*\| + \alpha \|{v}_{\tilde{x}_i} - \frac{1}{2^d} (\tilde{x}_i - x_*)\| \nonumber \\
\leq& \left( 1 - \frac{\alpha}{2^d} \right) \|\tilde{x}_i - x_*\| + \alpha \left( \frac{1}{8} \frac{1}{2^d} \|\tilde{x}_i - x_*\| + \frac{1}{2^{d/2}} \omega \right) \nonumber \\
=& \left( 1 - \frac{\alpha}{2^d} \frac{7}{8} \right) \|\tilde{x}_i - x_*\| + \alpha \frac{1}{2^{d/2}} \omega, \label{GA:e64}
\end{align}
where the second inequality holds by Lemma~\ref{GA:le10}.
By the assumptions $\tilde{x}_i \in \mathcal{B}(x_*, K_3 d^{-10} \|x_*\|)$, $\omega \leq \frac{K_1 \|x_*\| }{d^{16} 2^{d/2} }$, and using~\eqref{GA:e64}, we have $x_{i + 1} \in \mathcal{B}(x_*, K_3 d^{-10} \|x_*\|)$. In addition, using Lemma~\ref{GA:le14} yields that $\tilde{x}_{i + 1} = x_{i+1}$. Repeat the above steps yields that $x_i \in \mathcal{B}(x_*, K_3 d^{-10} \|x_*\|)$ and $\tilde{x}_i = x_i$ for all $ i \geq N$.

Using~\eqref{GA:e64} and $\alpha = K_4 \frac{2^d}{d^2}$, we have
\begin{align} \label{GA:e65}
\|x_{i + 1} - x_*\|\leq b_2 \|x_i - x_*\| + b_3 \frac{2^{d/2}}{d^2} \omega,
\end{align}
where $b_2 = 1 - 7 K_4 / (8 d^2)$ and $b_3$ is a universal constant. Repeatedly applying~\eqref{GA:e65} yields
\begin{align*}
\|x_{i + 1} - x_*\| \leq& b_2^{i + 1 - N} \|x_N - x_*\| + (b_2^{i-N} + b_2^{i - N - 1} + \cdots + 1) \frac{b_3 2^{d/2}}{d^2} \omega \\
\leq& b_2^{i + 1 - N} \|x_N - x_*\| + \frac{b_3 2^{d/2}}{(1 - b_2) d^2} \omega \\
\leq& b_2^{i + 1 - N} \|x_N - x_*\| + b_4 2^{d/2} \omega,
\end{align*}
where the last inequality follows from the definition of $b_2$, and $b_4$ is a universal constant. This finishes the proof for~\eqref{GA:e68:2}. Inequality~\eqref{GA:e74:2} follows from Lemma~\ref{GA:le19}.

This concludes the proof.

The remainder of the proof is devoted to prove the lemmas used in this section.

\subsection{Proof of Equation~\eqref{GA:e24}}

Our proof relies on $h_x$ being Lipschitz, as formalized by the lemma below, which is proven in Section~\ref{sec:prooflemma:h-lipschitz}:

\begin{lemma}
\label{lemma:h-lipschitz}
For any $x, y \notin \ball(0, K_0 \norm{\xo})$, where $K_0$ and $K_4$ are numerical constants,
\[
\norm{h_{x} - h_{y}}
\leq \frac{K_4d^2}{2^d} \norm{x-y}.
\]
\end{lemma}
By Lemma~\ref{lemma:h-lipschitz}, for all $t \in [0, 1]$ and $i > N$  (recall that by Lemma~\ref{lem:notinzeroball}, after at most $N$ steps, $x_i \neq \ball (0, K_0 \norm{\xo})$):
\begin{equation}
\label{lem:hxhatmhlip}
\|h_{\hat{x}_{i}} - h_{x_{i}}\| \leq \frac{K_4d^2}{2^d} \|\hat{x}_{i} - x_{i}\|,
\end{equation}
where $\hat{x}_{i} = x_{i} - t \alpha \stepdir{x_i}$.
Thus, we have that on $\eventnoisesmall$, for any $v_{\hat{x}_{i}} \in \partial f(\hat{x}_{i}) $  by Lemma \ref{lemma:vxtilde-hxtilde-subdifferential},
\begin{align}
\|v_{\hat{x}_{i}} - \stepdir{x_i}\|
\leq&
\|v_{\hat{x}_{i}} - h_{\hat{x}_{i}}\| + \|h_{\hat{x}_{i}} - h_{x_{i}}\| + \|h_{x_{i}} - \stepdir{x_i}\| \nonumber \\
\leq&
K_1 \frac{d^3 \sqrt{\epsilon}}{2^d} \max(\|\hat{x}_{i}\|, \|x_*\|) + \frac{\omega}{2^{d/2}} + \frac{K_4d^2}{2^d} \|\hat{x}_{i} - x_{i}\| + K_1 \frac{d^3 \sqrt{\epsilon}}{2^d} \max(\|x_{i}\|, \|x_*\|) + \frac{\omega}{2^{d/2}} \nonumber \\
\leq&
K_1 \frac{d^3 \sqrt{\epsilon}}{2^d} \max(\|x_{i}\| + \alpha \|\stepdir{x_i}\|, \|x_*\|) + \frac{K_4d^2}{2^d} \alpha \|\stepdir{x_i}\| + K_1 \frac{d^3 \sqrt{\epsilon}}{2^d} \max(\|x_{i}\|, \|x_*\|) + 2\frac{\omega}{2^{d/2}}  \nonumber \\
\leq&
K_1 \frac{d^3 \sqrt{\epsilon}}{2^d} \left(2 + \frac{\alpha d K}{2^d}  \right) \max(\|x_{i}\|, \|x_*\|) + \frac{K_4d^2}{2^d} \alpha \|\stepdir{x_i}\| + 2\frac{K_9/d^2}{2^{d}} \norm{\xo} \label{eta-minus-vtilde}
\end{align}
where the second inequality is from Lemma \ref{lemma:vxtilde-hxtilde-subdifferential} and Equation~\eqref{lem:hxhatmhlip},
and the fourth inequality is from~\eqref{GA:e23}
and the assumption $\frac{ \omega }{2^{-d/2} \|\xo\|_2} \leq K_9 / d^2$.

Combining \eqref{eta-minus-vtilde} and \eqref{lowerbound-vxtilde-with-noise}, we get that
\begin{align*}
\|v_{\hat{x}_{i}} - \tilde{v}_{x_{i}}\|
\leq&
\left(
\frac{5}{6}  + \alpha K_7 \frac{d^2}{2^d}
\right)
\norm{ \stepdir{x_i} },
\end{align*}
with the appropriate constants chosen sufficiently small.
This concludes the proof of Equation~\eqref{GA:e24}.

\subsection{Proof of Lemma~\ref{lem:notinzeroball}}
First suppose that $x_i \in \ball(0,2K_0 \norm{\xo})$.
We show that after a polynomial number of iterations $N$, we have that $x_{i+N} \notin \ball(0,2K_0 \norm{\xo})$.
Below, we prove that
\begin{align}
\label{eq:properstepzero}
\text{
$\innerprod{x}{\whcomm{}{v_{x}}} < 0$
and
$\whcomm{}{\norm{ v_x }} \geq \frac{1}{2^d 16\pi} \norm{\xo}$
for all
$x \in \mathcal{B}(0, 2K_0 \norm{\xo})$ \whcomm{}{and $v_x \in \partial f(x)$}}.
\end{align}
It follows that for any
$\whcomm{}{\tilde{x}_i} \in \ball(0, 2K_0 \norm{\xo})$,
\whcomm{}{$\tilde{x}_i$} and the next iterate produced by the algorithm,
\whcomm{}{$x_{i+1} = \tilde{x}_i - \alpha v_{\tilde{x}_i}$}, form an obtruse triangle.
As a consequence,
\whcomm{}{
\begin{align*}
\norm{\tilde{x}_{i + 1}}^2 = \norm{x_{i + 1}}^2
&\geq
\norm{\tilde{x}_i}^2 + \alpha^2 \norm{v_{\tilde{x}_i}}^2 \\
&\geq
\norm{\tilde{x}_i}^2 + \alpha^2 \frac{1}{(2^{d} 16\pi)^2} \norm{\xo}^2,
\end{align*}
}
where the last inequality follows from~\eqref{eq:properstepzero}.
Thus, the norm of the iterates \whcomm{}{$\tilde{x}_i$} will increase until after $\bigl(\frac{2 K_0 2^d 16 \pi}{\alpha} \bigr)^2$  iterations, we have $\whcomm{}{\tilde{x}_{i+N}} \notin \ball(0, 2K_0 \norm{\xo})$.

\whcomm{}{
Consider $\tilde{x}_i \notin \mathcal{B}(0, 2 K_0 \|x_*\|)$, and note that
\begin{align*}
\alpha \|{v}_{\tilde{x}_i}\| \leq \alpha \frac{d K}{2^d} \max(\|\tilde{x}_i\|, \|x_*\|) \leq \alpha \frac{16 \pi K d}{2^d}\|\tilde{x}_i\| \leq \frac{1}{2}\|\tilde{x}_i\|,
\end{align*}
where the first inequality follows from~\eqref{GA:e23}, the second inequality from $\|\tilde{x}_i\| \geq 2 K_0 \|x_*\|$, and finally the last inequality from our assumption on the sufficiently small step size $\alpha$. Therefore, from $x_{i + 1} = \tilde{x}_i - \alpha {v}_{\tilde{x}_i}$, we have that $t \tilde{x}_i + (1 - t) x_{i+1} \notin \mathcal{B}(0, K_0 \|x_*\|)$ for all $t \in [0, 1]$, which completes the proof.
}

%
%

\paragraph{Proof of~\eqref{eq:properstepzero}:}
It remains to prove~\eqref{eq:properstepzero}.
We start with proving $\innerprod{x}{\stepdir{x}} < 0$.
For brevity of notation, let $\Lambda_{z} = \prod_{i = d}^1 W_{i, +, z}$.
We have
\whcomm{}{
\begin{align*}
x^T \stepdir{x} =&
\innerprod{ \Lambda_x^T \Lambda_x x - \Lambda_x^T \Lambda_{x_*} x_* + \Lambda_x^T \eta }{x} \\
\leq&
\frac{13}{12}2^{-d} \|x\|^2 - \frac{1}{4 \pi} \frac{1}{2^d} \|x\| \|x_*\| +
\norm{x} \frac{\omega}{2^{d/2}} \\
\leq&
\|x\| \left( \frac{13}{12}2^{-d} \|x\| +
\frac{1/(16 \pi)}{2^{d}} \|x_*\|
- \frac{1}{4 \pi} \frac{1}{2^d} \|x_*\| \right) \\
\leq&
\|x\| \frac{1}{2^d} \left( 2\|x\| - \frac{3}{16 \pi} \|x_*\| \right).
\end{align*}
}
The first inequality follows from \eqref{piwipxxpiwipyy} and~\eqref{eq:boundGammax}, and the second inequality follows from our assumption on $\omega$.
Therefore, for any $x \in \mathcal{B}(0, \frac{1}{16\pi} \|x_*\|)$, 
\whcomm{}{
$
\innerprod{x}{ \stepdir{x} } < -\frac{1}{16 \pi 2^d} \|x\| \|x_*\| \leq 0,
$
}
as desired.

If $G(x)$ is differentiable at $x$, then $v_x = \tilde{v}_x$ and $\inner[]{x}{{v}_x} < 0$. If $G(x)$ is not differentiable at $x$, by equation~\eqref{subdifferential-convex-hull}, we have
\whcomm{}{
\begin{align}
x^T {v}_x =& x^T (c_1 v_1 + c_2 v_2 + \cdots + c_t v_t) \leq (c_1 + c_2 + \ldots + c_t) \|x\| \frac{1}{2^d} \left( 2\|x\| - \frac{3}{16 \pi} \|x_*\| \right) \nonumber \\
=& \|x\| \frac{1}{2^d} \left( 2\|x\| - \frac{3}{16 \pi} \|x_*\| \right) < -\frac{1}{16 \pi 2^d} \|x\| \|x_*\| \leq 0, \label{GA:e75}
\end{align}
}
for all $v_x \in \partial f(x)$.

\whcomm{}{
Using~\eqref{GA:e75} yields
$$
\|v_x\| = \max_{\|u\|=1} \inner[]{v_x}{u} \geq \inner[]{v_x}{-x/\|x\|} = -\frac{x^T {v}_x}{\|x\|} > \frac{1}{2^d 16\pi} \|x_*\|,
$$
which concludes the proof of~\eqref{eq:properstepzero}.
}

%
%

%


\subsection{Proof of Lemma~\ref{lem:noisesmall}}

Let $\Lambda_x = \PiWd$.
We have that
\[
\norm{\qbarx}^2
=
\norm{ \Lambda_x^t \eta}^2
\leq
\norm{\Lambda_x}^2
\norm{ P_{\Lambda_x} \eta}^2,
\]
where $P_{\Lambda_x}$ is a projector onto the span of $\Lambda_x$.
As a consequence, $\norm{ P_{\Lambda_x} \eta}^2$ is $\chi^2$-distributed random variable with $k$-degrees of freedom scaled by $\sigma/n$. A standard tail bound (see~\cite[p.~43]{boucheron_concentration_2013}) yields that, for any $\beta \geq k$,
\[
\PR{\norm{ P_{\Lambda_x} \eta}^2 \geq 4 \beta}
\leq
2e^{-\beta}.
\]
Next, we note that by applying Lemmas 13-14 from \cite[Proof of Lem.~15]{hand_global_2017})\footnote{The proof in that argument only uses the assumption of independence of subsets of rows of the weight matrices.}, with probability one, that the number of different matrices $\Lambda_x$ can be bounded as
\[
|\left\{
\Lambda_x | x \neq 0
\right\}|
=
|\left\{
\PiWd | x \neq 0
\right\}|
\leq
10^{d^2}(n_1^d n_2^{d-1}\ldots n_d)^k
\leq
(n_1^d n_2^{d-1}\ldots n_d)^{2k},
\]
where the second inequality holds for
$\log(10) \leq k/4 \log(n_1)$.
To see this, note that
$(n_1^d n_2^{d-1}\ldots n_d)^{k} \geq 10^{d^2}$
is implied by
$k( d\log(n_1)+ (d-1) \log(n_2) + \ldots \log(n_d))
\geq k d^2/4 \log(n_1) \geq d^2 \log(10)$.
Thus, by the union bound,
\[
\PR{\norm{ P_{\Lambda_x} \eta}^2
\leq
16k\log(n_1^d n_2^{d-1}\ldots n_d),
\text{ for all $x$}
}
\geq
1 - 2e^{- 2 k\log (n)},
\]
where $n = n_d$.
Recall from~\eqref{eq:boundGammax} that $\norm{\Lambda_x} \leq \frac{13}{12}$.
Combining this inequality with
$\norm{\qbarx}^2
\leq
\norm{\Lambda_x}^2
\norm{ P_{\Lambda_x} \eta}^2$ concludes the proof.


\subsection{Proof of Lemma \ref{lemma:Sepsst} } \label{sec:zeros-hxxo}

We now show that $\hx$ is away from zero outside of a neighborhood of $\xo$ and $-\rho_d \xo$.  We prove Lemma \ref{lemma:Sepsst} by establishing the following:

\begin{lemma} \label{lemma:Seps}
Suppose $64 d^6 \sqrt{\beta} \leq 1$.
Define
\[
\zetacheck_d := \sum_{i=0}^{d-1} \frac{\sin \thetacheck_{i}}{\pi} \left( \prod_{j=i+1}^{d-1} \frac{\pi - \thetacheck_{j} }{\pi} \right),
\]
where $\thetacheck_0= \pi$ and $\thetacheck_i = g(\thetacheck_{i-1})$.  If $x \in \setS_\beta$, then we have that either
\[
|\thetabar_0| \leq 32 d^4 \beta \quad \text{and} \quad |\|x\|_2 - \|\xo\|_2| \leq  132 d^6 \beta \|\xo\|_2
\]
or
\[
|\thetabar_0 - \pi|
\leq
8\pi d^4 \sqrt{\beta}  \quad \text{and} \quad  \left| \|x\|_2 - \|\xo\|_2 \zetacheck_d \right | \leq 200 d^7 \sqrt{\epsilon} \|\xo\|_2.
\]
In particular, we have
\begin{align}
\setS_\beta \subset \mathcal{B}(\xo, 5000 d^6 \beta \|\xo\|_2  ) \cup  \mathcal{B}( -\zetacheck_d \xo, 500 d^{11} \sqrt{\beta} \|\xo\|_2).
 \label{Seps-estimate-balls2}
\end{align}
Additionally, $\rho_d \to 1$ as $d \to \infty$.
\end{lemma}

\begin{proof}
Without loss of generality, let $\norm{\xo}
 =1$, $\xo = e_1$ and $\xhat = r \cos \thetabar_0 \cdot e_1 + r \sin \thetabar_0 \cdot e_2$ for $\thetabar_0 \in [0, \pi]$.  Let $x \in \setS_\beta$.

First we introduce some notation for convenience.  Let
\[
\xi = \Pipithetapii, \quad \zeta = \zetaterm, \quad r = \|x\|_2, \quad M = \max(r, 1).
\]
Thus, $\hx = - \frac{1}{2^d}\xi  \xohat + \frac{1}{2^d} (r- \zeta) \xhat$.
By inspecting the components of $\hx$, we have that $x\in \setS_\beta$ implies
\begin{align}
|-\xi + \cos \thetabar_0 ( r - \zeta)| \leq \beta M \label{Seps-e1term}\\
|\sin \thetabar_0 (r - \zeta)| \leq \beta M \label{Seps-e2term}
\end{align}

Now, we record several properties.  We have:
\begin{align}
\thetabar_{i} &\in [0, \pi/2] \text{ for } i \geq 1 \notag \\
\thetabar_{i} &\leq \thetabar_{i-1} \text{ for } i \geq 1 \notag\\
 |\xi| &\leq 1 \label{bound-c-one}\\
| \zeta| &\leq \frac{d}{\pi} \sin \theta_0 \label{zeta-bound-d-sin-theta} \\
\thetacheck_i &\leq \frac{3\pi}{i+3} \text{ for $i \geq 0$} \label{upper-bound-thetai}\\
\thetacheck_i &\geq \frac{\pi}{i+1} \text{ for $i \geq 0$} \label{lower-bound-thetai}\\
\xi = \prod_{i=0}^{d-1} \frac{\pi - \thetabar_i}{\pi} &\geq \frac{\pi - \thetabar_0}{\pi } d^{-3} \label{pi-theta-pi-d-cubed}\\
\thetabar_0 = \pi + O_1(\delta) &\Rightarrow \thetabar_i = \thetacheck_i + O_1(i \delta) \label{thetabar-thetacheck-bound}\\
\thetabar_0 = \pi + O_1(\delta) &\Rightarrow |\xi| \leq \frac{\delta}{\pi} \label{c-one-bound-thetabarzero}\\
\thetabar_0 = \pi + O_1(\delta) &\Rightarrow \zeta = \zetacheck_d + O_1(3 d^3 \delta) \text{ if } \frac{d^2 \delta}{\pi} \leq 1 \label{high-theta-zeta-zetabar-control}
\end{align}
We now establish \eqref{upper-bound-thetai}.  Observe $0< g(\theta) \leq \bigl( \frac{1}{3 \pi} + \frac{1}{\theta} \bigr)^{-1} =: \gtilde(\theta)$ for $\theta \in (0, \pi]$.  As $g$ and $\gtilde$ are monotonic increasing, we have $\thetacheck_i = g^{\circ i}(\thetacheck_0) = g^{\circ i}(\pi) \leq \gtilde^{\circ i}(\pi) = \bigl( \frac{i}{3 \pi} + \frac{1}{\pi} \bigr)^{-1} = \frac{3\pi}{i + 3}$.  Similarly, $g(\theta) \geq (\frac{1}{\pi} + \frac{1}{\theta})^{-1}$ implies that $\thetacheck_i \geq \frac{\pi}{i+1}$, establishing \eqref{lower-bound-thetai}.

We now establish \eqref{pi-theta-pi-d-cubed}.  Using \eqref{upper-bound-thetai} and $\thetabar_i \leq \thetacheck_i$, we have
\begin{align*}
\prod_{i=1}^{d-1} \Bigl(1 - \frac{\thetabar_i}{\pi} \Bigr) &\geq \prod_{i=1}^{d-1} \Bigl(1 - \frac{3}{i+3}  \Bigr)
\geq d^{-3},
\end{align*}
where the last inequality can be established by showing that the ratio of consecutive terms with respect to $d$ is greater for the product in the middle expression than for $d^{-3}$.

We establish \eqref{thetabar-thetacheck-bound} by using the fact that $|g'(\theta)| \leq 1$ for all $\theta \in [0, \pi]$ and using the same logic as for~\cite[Eq.~17]{hand_global_2017}.

We now establish \eqref{high-theta-zeta-zetabar-control}.  As $\thetabar_0 = \pi + O_1(\delta)$, we have $\thetabar_i = \thetacheck_i + O_1(i \delta)$.  Thus, if $\frac{d^2 \delta}{\pi}\leq 1$,
\[
\prod_{j=i+1}^{d-1} \frac{\pi - \thetabar_j}{\pi} = \prod_{j=i+1}^{d-1} \Bigl( \frac{\pi - \thetacheck_j}{\pi} + O_1(\frac{i \delta}{2 \pi} ) \Bigr) = \Bigl(\prod_{j=i+1}^{d-1}  \frac{\pi - \thetacheck_j}{\pi} \Bigr)  + O_1(d^2 \delta )
\]
So
\begin{align}
\zeta &= \sum_{i=0}^{d-1} \Bigl( \frac{\sin \thetacheck_i}{\pi} + O_1(\frac{i \delta}{\pi}) \Bigr)  \Bigl[ \Bigl(\prod_{j=i+1}^{d-1} \frac{\pi - \thetacheck_j}{\pi} \Bigr) +O_1(d^2 \delta)  \Bigr] \\
&= \zetacheck_d + O_1 \Bigl( d^2\delta / \pi + d^3 \delta /\pi + d^4 \delta^2/\pi   \Bigr) \\
&= \zetacheck_d + O_1( 3 d^3 \delta).
\end{align}
Thus \eqref{high-theta-zeta-zetabar-control} holds.

Next, we establish that $x \in \setS_\beta \Rightarrow r \leq 4d$, and thus $M \leq 4d$.  Suppose $r > 1$.  At least one of the following holds: $|\sin\thetabar_0| \geq 1/\sqrt{2}$ or $|\cos \thetabar_0| \geq 1/\sqrt{2}$. If $|\sin \thetabar_0| \geq 1/\sqrt{2}$ then  \eqref{Seps-e2term} implies that $|r - \zeta| \leq \sqrt{2} \beta r$.  Using \eqref{zeta-bound-d-sin-theta}, we get $ r \leq \frac{d/\pi}{1 - \sqrt{2} \beta} \leq d/2$ if $\beta < 1/4$.  If $| \cos \thetabar_0| \geq 1/\sqrt{2}$, then \eqref{Seps-e1term} implies that
$|r - \zeta| \leq \sqrt{2} (\beta r + |\xi|)$.   Using \eqref{bound-c-one}, \eqref{zeta-bound-d-sin-theta}, and $\beta<1/4$, we get $r \leq \frac{\sqrt{2} |\xi| +  \zeta}{1 - \sqrt{2}{\beta}} \leq \frac{d + \sqrt{2}}{1 - \sqrt{2} \beta} \leq 4 d$. Thus, we have $x \in S_\beta \Rightarrow r \leq 4d \Rightarrow M \leq 4d$.

Next, we establish that we only need to consider the small angle case ($\thetabar_0 \approx 0$) and the large angle case ($\thetabar_0 \approx \pi$), by considering the following three cases:
\begin{enumerate}
\item[] (Case I) $\sin \thetabar_0 \leq 16 d^4 \beta$: We have $\thetabar_0 = O_1(32 d^4 \beta)$ or $\thetabar_0 = \pi + O_1(32 d^4 \beta)$, as $32 d^4 \beta < 1$.

\item[] (Case II) $|r -\zeta| < \sqrt{\beta} M$: Applying case II to inequality~\eqref{Seps-e1term} yields $|\xi| \leq 2 \sqrt{\beta} M$.  Using \eqref{pi-theta-pi-d-cubed}, we get $\thetabar_0 = \pi + O_1(2 \pi d^3 \sqrt{\beta} M)$.

\item[] (Case III) $\sin \thetabar_0 > 16 d^4 \beta$ and $|r - \zeta| \geq \sqrt{\beta} M$:
   Finally, consider Case III.  By \eqref{Seps-e2term}, we have $|r - \zeta| \leq \frac{\beta M}{\sin\thetabar_0}$.
Using this inequality in~\eqref{Seps-e1term}, we have $|\xi| \leq \beta M + \frac{\beta M}{\sin \thetabar_0} \leq \frac{2 \beta M}{\sin \thetabar_0} \leq \frac{1}{8} d^{-4} M \leq \frac{1}{2} d^{-3}$, where the second to last inequality uses $\sin \thetabar_0 > 16 d^4 \beta$ and the last inequality uses $M \leq 4 d$.  By \eqref{pi-theta-pi-d-cubed},  we have $\frac{\pi - \thetabar_0}{\pi} d^{-3} \leq \xi \leq \frac{1}{2} d^{-3}$, which implies that $\thetabar_0 \geq \pi/2$.  Now, as $|r - \zeta| \geq \sqrt{\beta}M$, then by \eqref{Seps-e2term}, we have $|\sin\thetabar_0| \leq \sqrt{\beta}$.  Hence,  $\thetabar_0 = \pi + O_1(2 \sqrt{\beta})$, as $\thetabar_0 \geq \pi/2$ and as $\beta < 1$.
\end{enumerate}
At least one of the Cases I,II, or III hold. Thus, we see that it suffices to consider the small angle case $\thetabar_0 = O_1(32 d^4 \beta)$ or the large angle case $\thetabar_0 = \pi + O_1(8 \pi d^4 \sqrt{\beta})$.


\textbf{Small Angle Case}.  Assume $\thetabar_0 = O_1(\delta)$ with $\delta = 32 d^4 \beta$.  As $\thetabar_i \leq \thetabar_0 \leq \delta$ for all $i$, we have $1 \geq \xi \geq (1 - \frac{\delta}{\pi})^d = 1 + O_1(\frac{2 \delta d}{\pi})$ provided $\delta d/\pi \leq 1/2$ (which holds by our choice $\delta = 32 d^4 \beta$ by assumption $64 d^6 \sqrt{\beta} \leq 1$).   By \eqref{zeta-bound-d-sin-theta}, we also have $\zeta = O_1( \frac{d}{\pi} \delta)$.  By \eqref{Seps-e1term},  we have
\[
|-\xi + \cos \thetabar_0 ( r - \zeta)| \leq \beta M.
\]
Thus, as $\cos \thetabar_0 = 1 + O_1(\thetabar_0^2/2) = 1 + O_1(\delta^2/2)$,
\[
-\Bigl( 1 + O_1(\frac{2 \delta d}{\pi}) \Bigr) + (1 + O_1(\frac{2 \delta d}{\pi}) )(r + O_1(\frac{\delta d}{\pi})) = O_1(4 d \beta),
\]
and $r \leq M\leq 4 d$ (shown above) provides,
\begin{align}
r-1 &= O_1(4d \beta + \frac{2 \delta d}{\pi} + \frac{\delta d}{\pi} + \frac{2 \delta d}{\pi}4d + \frac{2 \delta^2 d^2}{\pi^2} )\\
&= O_1(4 \beta d + 4 \delta d^2).
\end{align}
By plugging in that $\delta = 32 d^4 \beta$, we have that $r-1 = O_1(132 d^6 \beta)$, where we have used that $\frac{32 d^5 \beta}{\pi} \leq 1/2$.

\textbf{Large Angle Case}.  Assume $\theta_0 = \pi + O_1(\delta)$ where $\delta = 8 \pi d^4 \sqrt{\beta}$.  By \eqref{c-one-bound-thetabarzero} and \eqref{high-theta-zeta-zetabar-control}, we have $\xi = O_1(\delta/\pi)$, and we have $\zeta = \zetacheck_d + O_1(3 d^3 \delta)$ if $8 d^6 \sqrt{\beta} \leq 1$.
By \eqref{Seps-e1term}, we have
\[
|-\xi +  \cos \theta_0 (r - \zeta) | \leq \beta M,
\]
so, as $\cos \theta_0 = 1 - O_1(\theta_0^2/2)$,
\[
O_1(\delta/\pi) + (1 + O_1(\delta^2/2))(r - \zetacheck_d + O_1(3 d^3 \delta)) = O_1(\beta M),
\]
and thus, using $r \leq 4 d$, $\zetacheck_d \leq d$, and $\delta = 8 \pi d^4 \sqrt{\beta} \leq 1$,
\begin{align}
r - \zetacheck_d &= O_1(\beta M + \delta/\pi + 3 d^3 \delta + \frac{5}{2} \delta^2 d + \frac{3}{2} d^3 \delta^3) \\
&= O_1 \Bigl(4 \beta d + \delta ( \frac{1}{\pi} + 3 d^3 + \frac{5}{2} d + \frac{3}{2} d^3) \Bigr) \\
&= O_1(200 d^7 \sqrt{\beta})
\end{align}

To conclude the proof of \eqref{Seps-estimate-balls2}, we use the fact that
\[
\|x - \xo\|_2 \leq  \bigl| \|x\|_2 - \|\xo\|_2   \bigr| + ( \|\xo\|_2 +  \bigl| \|x\|_2 - \|\xo\|_2   \bigr| ) \thetabar_0.
\]
This fact simply says that if a 2d point is known to have magnitude within $\Delta r$ of some $r$ and is known to be within angle $\Delta \theta$ from $0$, then its Euclidean distance to the point of polar coordinates $(r,0)$ is no more than $\Delta r + (r +  \Delta r) \Delta \theta$.

Finally, we establish that $\rho_d\to 1$ as $d \to \infty$.  Note that $\rho_{d+1} = (1 - \frac{\thetacheck_d}{\pi}) \rho_d + \frac{\sin \thetacheck_d}{\pi}$  and $\rho_0 = 0$.  It suffices to show $\rhotilde_d \to 0$, where $\rhotilde_d := 1 - \rho_d$.  The following recurrence relation holds: $\rhotilde_d = (1 - \frac{\thetacheck_{d-1}}{\pi}) \rhotilde_{d-1} + \frac{\thetacheck_{d-1} - \sin \thetacheck_{d-1}}{\pi}$, with $\rhotilde_0=1$.  Using the recurrence formula \cite[Eq.~(15)]{hand_global_2017}  and the fact that $\thetacheck_0=\pi$, we get that
\begin{align}
\rhotilde_d =  \sum_{i=1}^d \frac{\thetacheck_{i-1} - \sin \thetacheck_{i-1}}{\pi} \prod_{j=i+1}^d \bigl(1 - \frac{\thetacheck_{j-1}}{\pi} \bigr)
\end{align}
using \eqref{lower-bound-thetai}, we have that
\begin{align*}
\prod_{j=i+1}^d \Bigl(1 - \frac{\thetacheck_{j-1}}{\pi} \Bigr) &\leq \prod_{j=i+1}^d \Bigl(1 - \frac{1}{j} \Bigr)
= \exp \Bigl( -\sum_{j = i+1}^d \frac{1}{j}  \Bigr) \leq \exp \Bigl( - \int_{i+1}^{d+1} \frac{1}{s} ds  \Bigr) = \frac{i+1}{d+1}
\end{align*}
Using \eqref{upper-bound-thetai} and the fact that $\thetacheck_{i-1} - \sin \thetacheck_{i-1} \leq \thetacheck_{i-1}^3 / 6$, we have that $\rhotilde_d \leq \sum_{i=1}^d \frac{\thetacheck_{i-1}^3}{6\pi}\cdot  \frac{i+1}{d+1} \to 0$ as $d \to \infty$.

\end{proof}

\subsection{Proof of
Lemma \ref{GA:le14}}

Consider the function
\[
\feta(x) = \fzero(x) - \langle G(x) - G(\xo), \eta \rangle,
\]
and note that $f(x) = \feta(x) + \norm{\eta}^2$.
Consider $x \in \mathcal{B}(\phi_d x_*, \varphi  \|x_*\|)$, for a $\varphi$ that will be specified later.
Note that
\begin{align*}
\left|\innerprod{G(x) - G(x_\ast)}{\eta}\right|
&\leq
|\innerprod{\PiWd x}{\eta}|
+
|\innerprod{\PiWdo \xo}{\eta}| \\
&=
|\innerprod{x}{(\PiWd)^t\eta}|
+
|\innerprod{\xo}{(\PiWdo)^t\eta}| \\
&\leq (\norm{x} + \norm{\xo}) \frac{\omega}{2^{d/2}}\\
&\leq (\varphi \norm{\xo} + \norm{\xo}) \frac{\omega}{2^{d/2}},
\end{align*}
where the second inequality holds on the event $\eventnoisesmall$, by Lemma~\ref{lem:noisesmall}, and the last inequality holds by our assumption on $x$.
Thus, for $x \in \mathcal{B}(\phi_d x_*, \varphi  \|x_*\|)$
\begin{align}
\feta(x)
\leq&
\mathbb{E}\fzero(x) + |\fzero(x) - \mathbb{E}\fzero(x)| +
\left|\innerprod{G(x) - G(x_\ast)}{\eta}\right|
\nonumber \\
\leq&
\frac{1}{2^{d+1}} \left( \phi_d^2 - 2 \phi_d + \frac{10}{K_2^3} d \varphi \right) \|x_*\|^2 + \frac{1}{2^{d+1}} \|x_*\|^2 \nonumber \\
&+
\frac{\epsilon (1 + 4 \epsilon d)}{2^d} \|x\|^2 + \frac{\epsilon(1 + 4 \epsilon d) + 48 d^3 \sqrt{\epsilon}}{2^{d+1}} \|x\| \|x_*\| + \frac{\epsilon (1 + 4 \epsilon d)}{2^d} \|x_*\|^2 \nonumber \\
&+ (\varphi \norm{\xo} + \norm{\xo}) \frac{\omega}{2^{d/2}} \nonumber \\
\leq&
\frac{1}{2^{d+1}} \left( \phi_d^2 - 2 \phi_d + \frac{10}{K_2^3} d \varphi \right) \|x_*\|^2 + \frac{1}{2^{d+1}} \|x_*\|^2 \nonumber \\
&+
\frac{\epsilon (1 + 4 \epsilon d)}{2^d} (\phi_d + \varphi)^2 \|x_*\|^2 + \frac{\epsilon(1 + 4 \epsilon d) + 48 d^3 \sqrt{\epsilon}}{2^{d+1}} (\phi_d + \varphi)\|x_*\|^2 + \frac{\epsilon (1 + 4 \epsilon d)}{2^d} \|x_*\|^2  \nonumber \\
&+
(\varphi \norm{\xo} + \norm{\xo}) \frac{\omega}{2^{d/2}} \nonumber \\
\leq&
\frac{ \norm{x_*}^2 }{2^{d+1}}
\left(
1+
\phi_d^2 - 2 \phi_d + \frac{10}{K_2^3} d \epsilon
+
68 d^2 \sqrt{\epsilon}
\right)
+
(\varphi \norm{\xo} + \norm{\xo}) \frac{\omega}{2^{d/2}} \label{GA:e42}
\end{align}
where the last inequality follows from $\epsilon < \sqrt{\epsilon}$, $\rho_d \leq 1$, $4 \epsilon d < 1$, $\varphi < 1$ and assuming $\varphi = \epsilon$.

Similarly, we have that for any $y \in \mathcal{B}(- \phi_d x_*, \varphi  \|x_*\|)$
\begin{align}
f_\eta(y) \geq& \mathbb{E}[f(y)] - |f(y) - \mathbb{E}[f(y)]|
-
\left|\innerprod{G(x) - G(x_\ast)}{\eta}\right|
  \nonumber \\
\geq& \frac{1}{2^{d+1}} \left(\phi_d^2 - 2 \phi_d \rho_d - 10 d^3 \varphi \right) \|x_*\|^2 + \frac{1}{2^{d+1}} \|x_*\|^2 \nonumber \\
&- \left( \frac{\epsilon (1 + 4 \epsilon d)}{2^d} \|y\|^2 + \frac{\epsilon(1 + 4 \epsilon d) + 48 d^3 \sqrt{\epsilon}}{2^{d+1}} \|y\| \|x_*\| + \frac{\epsilon (1 + 4 \epsilon d)}{2^d} \|x_*\|^2 \right) \nonumber \\
&- (\varphi \norm{\xo} + \norm{\xo}) \frac{\omega}{2^{d/2}} \nonumber \\
\geq&
\frac{\|x_*\|^2}{2^{d+1}} \left(1 + \phi_d^2 - 2 \phi_d \rho_d - 10 d^3 \varphi
- 68 d^2 \sqrt{\epsilon}
\right) -(\varphi \norm{\xo} + \norm{\xo}) \frac{\omega}{2^{d/2}}
\label{GA:e43}
\end{align}
Using $\epsilon < \sqrt{\epsilon}$, $\rho_d \leq 1$, $4 \epsilon d < 1$, $\varphi < 1$ and assuming $\varphi = \epsilon$, the right side of~\eqref{GA:e42} is smaller than the right side of~\eqref{GA:e43} if
\begin{equation} \label{GA:e44}
\varphi = \epsilon \leq \left(\frac{\phi_d - \rho_d \phi_d  - 13 \|\etabar\|_2}{\left( 125 + \frac{5}{K_2^3} \right) d^3} \right)^2.
\end{equation}
We can establish that:
\begin{lemma}\label{lemma:rho-d}
For all $d\geq 2$, that
\[
1/\left(K_1(d + 2)^2\right) \leq 1 - \rho_d \leq 250/(d + 1).
\]
\end{lemma}
Thus, it suffices to have $\varphi = \epsilon = \frac{K_3}{d^{10}}$ and $13 \|\etabar\|_2 \leq \frac{K_9}{d^2} \leq \frac{1}{2} \frac{K_2}{K_1 (d+2)^2}$ for an appropriate universal constant $K_9$, and for an appropriate universal constant $K_3$.  

\subsection{Proof of Lemma \ref{lemma:rho-d}}
It holds that
\begin{align}
&\|x - y\| \geq 2 \sin (\theta_{x, y} / 2) \min(\|x\|, \|y\|), &\forall x, y \label{GA:e4} \\
&\sin(\theta / 2) \geq \theta / 4, &\forall \theta \in [0, \pi] \label{GA:e5} \\
&\frac{d}{d \theta} g(\theta) \in [0, 1] &\forall \theta \in [0, \pi] \label{GA:e13} \\
&\log(1+x) \leq x &\forall x \in [-0.5, 1] \label{GA:e47} \\
&\log(1-x) \geq -2 x &\forall x \in [0, 0.75] \label{GA:e48}
\end{align}
where $\theta_{x, y} = \angle(x, y)$.
We recall the results (36), (37), and (50) in \citep{hand_global_2017}:
\begin{align*}
&\check{\theta}_i \leq \frac{3 \pi}{i + 3}\;\;\;\;\;  \;\;\; \hbox{ and }\;\;\;\;\;  \;\;\; \check{\theta}_i \geq \frac{\pi}{i + 1}\;\;\;\;\; \forall i \geq 0 \\
&1 - \rho_d = \prod_{i = 1}^{d - 1} \left( 1 - \frac{\check{\theta}_{i}}{\pi} \right) + \sum_{i = 1}^{d-1} \frac{\check{\theta}_{i} - \sin \check{\theta}_{i}}{\pi} \prod_{j = i+1}^{d-1} \left( 1 - \frac{\check{\theta}_{j}}{\pi} \right).
\end{align*}
Therefore, we have for all $0 \leq i \leq d - 2$,
\begin{align*}
&\prod_{j = i+1}^{d-1} \left( 1 - \frac{\check{\theta}_{j}}{\pi} \right)
\leq \prod_{j = i+1}^{d-1} \left( 1 - \frac{1}{j + 1} \right) = e^{\sum_{j = i + 1}^{d - 1} \log\left(1 - \frac{1}{j + 1}\right)}
\leq e^{- \sum_{j = i + 1}^{d - 1}  \frac{1}{j + 1}} \leq e^{- \int_{i + 1}^d \frac{1}{s + 1} d s} = \frac{i + 2}{d + 1}, \\
&\prod_{j = i+1}^{d-1} \left( 1 - \frac{\check{\theta}_{j}}{\pi} \right)
\geq
\prod_{j = i+1}^{d-1} \left( 1 - \frac{3}{j + 3} \right) = e^{\sum_{j = i + 1}^{d - 1} \log\left(1 - \frac{3}{j + 3}\right)}
\geq e^{- \sum_{j = i + 1}^{d - 1}  \frac{6}{j + 3}} \geq e^{- \int_{i}^{d - 1} \frac{6}{s + 3} d s} = \left(\frac{i + 3}{d + 2}\right)^6,
\end{align*}
where the second and the fifth inequalities follow from~\eqref{GA:e47} and~\eqref{GA:e48} respectively.
Since $\pi^3 / (12 (i + 1)^3) \leq \check{\theta}_{i}^3 / 12 \leq \check{\theta}_{i} - \sin \check{\theta}_{i} \leq \check{\theta}_{i}^3 / 6 \leq 27 \pi^3 / (6 (i + 3)^3)$, we have that for all $d \geq 3$
\begin{align*}
1 - \rho_d \leq& \frac{2}{d + 1} + \sum_{i = 1}^{d - 1} \frac{27 \pi^3}{6 (i + 3)^3} \frac{i + 2}{d + 1} \leq \frac{2}{d + 1} + \frac{3 \pi^5}{4 (d + 1)} \leq \frac{250}{d + 1}
\end{align*}
and
\begin{align*}
1 - \rho_d \geq& \left(\frac{3}{(d+2)}\right)^6 + \sum_{i = 1}^{d - 1} \frac{\pi^3}{12 (i + 3)^3} \left( \frac{i + 3}{d + 2} \right)^6 \geq \frac{1}{K_1 (d + 2)^2},
\end{align*}
where we use $\sum_{i = 4}^\infty \frac{1}{i^2} \leq \frac{\pi^2}{6}$ and $\sum_{i = 1}^n i^3 = O(n^4)$.

\subsection{\label{sec:prooflemma:h-lipschitz}Proof of Lemma~\ref{lemma:h-lipschitz}}
To establish Lemma \ref{lemma:h-lipschitz}, we prove the following:

\begin{lemma}
\label{GA:le4}
For all $x, y \neq 0$,
\begin{equation*}
\|h_{x} - h_{y}\| \leq \left(\frac{1}{2^d} + \frac{6 d + 4 d^2}{\pi 2^d} \max\left( \frac{1}{\|x\|}, \frac{1}{\|y\|} \right) \|x_*\|\right) \|x - y\|
\end{equation*}
\end{lemma}
Lemma~\ref{lemma:h-lipschitz} follows by noting that if $x, y \notin \mathcal{B}(0, r \|x_*\|)$, then $\|h_{x} - h_{y}\| \leq \left(\frac{1}{2^d} + \frac{6 d + 4 d^2}{\pi r 2^d}\right) \|x - y\|$.

\begin{proof}[Proof of Lemma~\ref{GA:le4}]
For brevity of notation, let $\zeta_{j, z} = \prod_{i = j}^{d - 1} \frac{\pi - \bar{\theta}_{i, z}}{\pi}$.
Combining \eqref{GA:e4} and~\eqref{GA:e5} gives $|\bar{\theta}_{0, x} - \bar{\theta}_{0, y}| \leq 4 \max\left( \frac{1}{\|x\|}, \frac{1}{\|y\|} \right) \|x - y\|$.
Inequality~\eqref{GA:e13} implies
$|\bar{\theta}_{i, x} - \bar{\theta}_{i, y}| \leq |\bar{\theta}_{j, x} - \bar{\theta}_{j, y}|$ for all $i \geq j$.
It follows that
\begin{align}
\|h_{x} - h_{y}\| \leq& \frac{1}{2^d}\|x - y\| + \frac{1}{2^d} \underbrace{\left| \zeta_{0, x} - \zeta_{0, y} \right|}_{T_1} \|x_*\| \nonumber \\
+& \frac{1}{2^d} \underbrace{ \left| \sum_{i = 0}^{d - 1} \frac{\sin \bar{\theta}_{i, x}}{\pi} \zeta_{i + 1, x} \hat{x} - \sum_{i = 0}^{d - 1} \frac{\sin \bar{\theta}_{i, y}}{\pi} \zeta_{i + 1, y} \hat{y} \right| }_{T_2} \|x_*\|. \label{GA:e50}
\end{align}
By Lemma~\ref{GA:le17}, we have
\begin{align} \label{GA:e53}
T_1 \leq \frac{d}{\pi} |\bar{\theta}_{0, x} - \bar{\theta}_{0, y}| \leq \frac{4 d}{\pi}\max\left( \frac{1}{\|x\|}, \frac{1}{\|y\|} \right) \|x - y\|.
\end{align}
Additionally, it holds that
\begin{align}
T_2 =& \left| \sum_{i = 0}^{d - 1} \frac{\sin \bar{\theta}_{i, x}}{\pi} \zeta_{i + 1, x} \hat{x} - \frac{\sin \bar{\theta}_{i, x}}{\pi} \zeta_{i + 1, x} \hat{y} + \frac{\sin \bar{\theta}_{i, x}}{\pi} \zeta_{i + 1, x} \hat{y} - \sum_{i = 0}^{d - 1} \frac{\sin \bar{\theta}_{i, y}}{\pi} \zeta_{i + 1, y} \hat{y} \right| \nonumber \\
\leq& \frac{d}{\pi} \|\hat{x} - \hat{y}\| + \underbrace{ \left| \sum_{i = 0}^{d - 1} \frac{\sin \bar{\theta}_{i, x}}{\pi} \zeta_{i + 1, x} - \sum_{i = 0}^{d - 1} \frac{\sin \bar{\theta}_{i, y}}{\pi} \zeta_{i + 1, y} \right|}_{T_3} . \label{GA:e51}
\end{align}
We have
\begin{align}
T_3 \leq& \sum_{i = 0}^{d - 1} \left[ \left| \frac{\sin \bar{\theta}_{i, x}}{\pi} \zeta_{i + 1, x} -\frac{\sin \bar{\theta}_{i, x}}{\pi} \zeta_{i + 1, y} \right|\right. + \left.\left| \frac{\sin \bar{\theta}_{i, x}}{\pi} \zeta_{i + 1, y} - \frac{\sin \bar{\theta}_{i, y}}{\pi} \zeta_{i + 1, y} \right| \right] \nonumber \\
\leq& \sum_{i = 0}^{d - 1} \left[ \frac{1}{\pi} \left( \frac{d - i  - 1}{\pi} \left|\bar{\theta}_{i-1, x} - \bar{\theta}_{i-1, y}\right| \right) + \frac{1}{\pi} |\sin \bar{\theta}_{i, x} - \sin \bar{\theta}_{i, y}|  \right] \nonumber \\
\leq& \frac{d^2}{\pi} |\bar{\theta}_{0, x} - \bar{\theta}_{0, y}| \leq \frac{4 d^2}{\pi} \max\left(\frac{1}{\|x\|}, \frac{1}{\|y\|}\right) \|x - y\|. \label{GA:e52}
\end{align}
Using~\eqref{GA:e4} and~\eqref{GA:e5} and noting $\|\hat{x} - \hat{y}\| \leq \theta_{x, y}$ yield
\begin{equation} \label{GA:e15}
\|\hat{x} - \hat{y}\| \leq \theta_{x, y} \leq 2 \max\left( \frac{1}{\|x\|}, \frac{1}{\|y\|} \right) \|x - y\|.
\end{equation}
Finally, combining~\eqref{GA:e50}, \eqref{GA:e53}, \eqref{GA:e51}, \eqref{GA:e52} and~\eqref{GA:e15} yields the result.
\end{proof}

\begin{lemma} \label{GA:le17}
Suppose $a_i, b_i \in [0, \pi]$ for $i = 1, \ldots, k$, and $|a_i - b_i| \leq |a_j - b_j|, \forall i \geq j$. Then it holds that
\begin{equation*}
\left|\prod_{i = 1}^k \frac{\pi - a_i}{\pi} - \prod_{i = 1}^k \frac{\pi - b_i}{\pi}\right| \leq \frac{k}{\pi} |a_1 - b_1|.
\end{equation*}
\end{lemma}
\begin{proof}
Prove by induction. It is easy to verify that the inequality holds if $k = 1$.
Suppose the inequality holds with $k = t - 1$. Then
\begin{align*}
\left|\prod_{i = 1}^{t} \frac{\pi - a_i}{\pi} - \prod_{i = 1}^t \frac{\pi - b_i}{\pi}\right| \leq& \left|\prod_{i = 1}^{t} \frac{\pi - a_i}{\pi} - \frac{\pi - a_t}{\pi} \prod_{i = 1}^{t-1} \frac{\pi - b_i}{\pi}\right| \\
&+ \left|\frac{\pi - a_t}{\pi} \prod_{i = 1}^{t-1} \frac{\pi - b_i}{\pi} - \prod_{i = 1}^t \frac{\pi - b_i}{\pi}\right| \\
\leq& \frac{t-1}{\pi} |a_1 - b_1| + \frac{1}{\pi} |a_t - b_t| \leq \frac{t}{\pi} |a_1 - b_1|.
\end{align*}
\end{proof}

\subsection{Proof of Lemma~\ref{GA:le10}}

We first need Lemmas~\ref{GA:le20}, ~\ref{GA:le19} and~\ref{GA:le8}.


\begin{lemma} \label{GA:le20}
Suppose $W \in \mathbb{R}^{n \times k}$ satisfies the WDC with constant $\epsilon$. Then for any $x, y \in \mathbb{R}^k$, it holds that
\begin{equation*}
\|W_{+, x} x - W_{+, y} y\| \leq \left( \sqrt{\frac{1}{2} + \epsilon} + \sqrt{2(2\epsilon + \theta)} \right) \|x - y\|,
\end{equation*}
where $\theta = \angle(x, y)$.
\end{lemma}
\begin{proof}
We have
\begin{align}
&\|W_{+, x} x - W_{+, y} y\| \leq \|W_{+, x} x - W_{+, x} y\| + \|W_{+, x} y - W_{+, y} y\| \nonumber \\
=& \|W_{+, x} (x - y)\| + \|(W_{+, x} - W_{+, y}) y\| \leq \|W_{+, x}\| \|x - y\| + \|(W_{+, x} - W_{+, y}) y\|. \label{GA:e33}
\end{align}
By WDC assumption, we have
\begin{align}
\|W_{+, x}^T (W_{+, x} - W_{+, y})\| \leq& \left\|W_{+, x}^T W_{+, x} - I  / 2\right\| + \left\|W_{+, x}^T W_{+, y} - Q_{x, y}\right\| + \left\| Q_{x, y} -  I / 2 \right\| \nonumber \\
\leq& 2 \epsilon + \theta. \label{GA:e34}
\end{align}
We also have
\begin{align}
&\|(W_{+, x} - W_{+, y}) y\|^2 = \sum_{i = 1}^n (1_{w_i \cdot x > 0} - 1_{w_i \cdot y > 0})^2 (w_i \cdot y)^2 \nonumber \\
\leq& \sum_{i = 1}^n (1_{w_i \cdot x > 0} - 1_{w_i \cdot y > 0})^2 ((w_i \cdot x)^2 + (w_i \cdot y)^2 - 2 (w_i \cdot x) (w_i \cdot y) ) \nonumber \\
=& \sum_{i = 1}^n (1_{w_i \cdot x > 0} - 1_{w_i \cdot y > 0})^2 (w_i \cdot (x - y) )^2 \nonumber \\
=& \sum_{i = 1}^n 1_{w_i \cdot x > 0} 1_{w_i \cdot y \leq 0} (w_i \cdot (x - y))^2 + \sum_{i = 1}^n 1_{w_i \cdot x \leq 0} 1_{w_i \cdot y > 0} (w_i \cdot(x - y))^2 \nonumber \\
=& (x - y)^T W_{+, x}^T (W_{+, x} - W_{+, y}) (x - y) + (x - y)^T W_{+, y}^T (W_{+, y} - W_{+, x}) (x - y) \nonumber \\
\leq& 2 (2 \epsilon + \theta) \|x - y\|^2. \;\;\; \hbox{ (by~\eqref{GA:e34})} \label{GA:e35}
\end{align}
Combining~\eqref{GA:e33},~\eqref{GA:e35}, and $\|W_{i, +, x}\|^2 \leq 1/2 + \epsilon$ given in~\cite[(10)]{hand_global_2017} yields the result.
\end{proof}


\begin{lemma} \label{GA:le19}
Suppose $x \in \mathcal{B}(x_*, d \sqrt{\epsilon} \|x_*\|)$, and
the WDC
holds with $\epsilon < 1/ (200)^4 / d^6$. Then it holds that
\begin{equation*}
\left\|\prod_{i = j}^1 W_{i, +, x} x - \prod_{i = j}^1 W_{i, +, x_*} x_*\right\| \leq \frac{1.2}{2^{\frac{j}{2}}} \|x - x_*\|.
\end{equation*}
\end{lemma}
\begin{proof}
In this proof, we denote $\theta_{i, x, x_*}$ and $\bar{\theta}_{i, x, x_*}$ by $\theta_i$ and $\bar{\theta}_{i}$ respectively.
Since $x \in \mathcal{B}(x_*, d \sqrt{\epsilon} \|x_*\|)$, we have
\begin{equation} \label{GA:e40}
\bar{\theta}_{i} \leq \bar{\theta}_{0} \leq 2 d \sqrt{\epsilon}.
\end{equation}
By~\cite[(14)]{hand_global_2017}, we also have
$|\theta_{i} - \bar{\theta}_{i}| \leq 4 i \sqrt{\epsilon} \leq 4 d \sqrt{\epsilon}$.
It follows that
\begin{align}
2 \sqrt{\theta_i + 2 \epsilon} \leq& 2 \sqrt{\bar{\theta}_i + 4 d \sqrt{\epsilon} + 2\epsilon} \leq 2 \sqrt{2 d \sqrt{\epsilon} + 4 d \sqrt{\epsilon} + 2\epsilon} \nonumber \\
\leq& 2 \sqrt{8 d \sqrt{\epsilon}} \leq \frac{1}{30 d}. \hbox{ (by the assumption on $\epsilon$)} \label{GA:e39}
\end{align}
Note that $\sqrt{1 + 2 \epsilon} \leq 1 + \epsilon \leq 1 + \sqrt{d\sqrt{\epsilon}}$. We have
\begin{align*}
\prod_{i = d - 1}^0 \left( \sqrt{1+2\epsilon} + 2 \sqrt{ \theta_i + 2 \epsilon} \right) \leq& \left(1 + 7 \sqrt{d\sqrt{\epsilon}}\right)^d \leq 1 + 14 d \sqrt{d\sqrt{\epsilon}} \leq \frac{107}{100} < 1.2,
\end{align*}
where the second inequality is from that $(1+x)^d \leq 1 + 2dx$ if $0 < x d < 1$.
Combining the above inequality with Lemma~\ref{GA:le20} yields
\begin{align*}
\left\|\prod_{i = j}^1 W_{i, +, x}x - \prod_{i = j}^1 W_{i, +, x_*} x_*\right\| \leq \prod_{i = j - 1}^0 \left( \sqrt{\frac{1}{2}+\epsilon} + \sqrt{2} \sqrt{ \theta_i + 2 \epsilon} \right) \|x - x_*\| \leq \frac{1.2}{2^{\frac{j}{2}}} \|x - x_*\|.
\end{align*}
\end{proof}


\begin{lemma}\label{GA:le8}
Suppose $x \in \mathcal{B}(x_*, d \sqrt{\epsilon} \|x_*\|)$, and
the WDC
holds with $\epsilon < 1/ (200)^4 / d^6$.  Then it holds that
\begin{equation*}
\left(\prod_{i = d}^1 W_{i, +, x}\right)^T\left[\left(\prod_{i = d}^1 W_{i, +, x}\right) x - \left(\prod_{i = d}^1 W_{i, +, x_*}\right) x_*\right] = \frac{1}{2^d} (x - x_*) + \frac{1}{2^d} \frac{1}{16} \|x - x_*\| O_1(1).
\end{equation*}
\end{lemma}
\begin{proof}
For brevity of notation, let $\Lambda_{j, k, z} = \prod_{i = j}^k W_{i, +, z}$.
We have
\begin{align}
&\Lambda_{d, 1, x}^T\left(\Lambda_{d, 1, x} x - \Lambda_{d, 1, x_*} x_*\right) \nonumber \\
=& \Lambda_{d, 1, x}^T\left[\Lambda_{d, 1, x} x - \sum_{j = 1}^d \left(\Lambda_{d, j, x} \Lambda_{j-1, 1, x_*} x_*\right)\right. + \left.\sum_{j = 1}^d \left(\Lambda_{d, j, x} \Lambda_{j-1, 1, x_*} x_*\right) - \Lambda_{d, 1, x_*} x_*\right] \nonumber \\
=& \underbrace{\Lambda_{d, 1, x}^T \Lambda_{d, 1, x} (x - x_*)}_{T_1} + \underbrace{\Lambda_{d, 1, x}^T \sum_{j = 1}^d \Lambda_{d, j + 1, x} \left( W_{j, +, x} - W_{j, +, x_*} \right) \Lambda_{j-1, 1, x_*} x_*}_{T_2}. \label{GA:e37}
\end{align}
For $T_1$, we have
\begin{align} \label{GA:e54}
T_1 = \frac{1}{2^d} (x - x_*) + \frac{4 d} {2^d} \|x - x_*\| O_1(\epsilon). \;\; \hbox{ \cite[(10)]{hand_global_2017}}
\end{align}
For $T_2$, we have
\begin{align}
T_2 =& O_1(1) \sum_{j = 1}^d \left( \frac{1}{2^{d - \frac{j}{2}}} + \frac{(4d - 2j)\epsilon}{2^{d - \frac{j}{2}}} \right) \left\|(W_{j, +, x} - W_{j, +, x_*}) \Lambda_{j-1, 1, x_*}x_* \right\| \nonumber \\
=& O_1(1) \sum_{j = 1}^d \left( \frac{1}{2^{d - \frac{j}{2}}} + \frac{(4d - 2j)\epsilon}{2^{d - \frac{j}{2}}} \right) \left\|(\Lambda_{j-1, 1, x} x - \Lambda_{j-1, 1, x_*} x_*) \right\| \sqrt{2 (\theta_{i, x, x_*} + 2 \epsilon)} \nonumber \\
=& O_1(1) \sum_{j = 1}^d \left( \frac{1}{2^{d - \frac{j}{2}}} + \frac{(4d - 2j)\epsilon}{2^{d - \frac{j}{2}}} \right) \frac{1.2}{2^{\frac{j}{2}}} \|x - x_*\| \frac{1}{ 30 \sqrt{2} d} \nonumber \\
=& \frac{1}{16} \frac{1}{2^d} \|x - x_*\| O_1(1). \label{GA:e55}
\end{align}
where the first equation is by~\cite[(10)]{hand_global_2017}; the second equation is by~\eqref{GA:e35}; the third equation is by Lemma~\ref{GA:le19} and~\eqref{GA:e39}. The result follows from~\eqref{GA:e37}, \eqref{GA:e54} and~\eqref{GA:e55}.
\end{proof}

Now, we are ready to prove Lemma~\eqref{GA:le10}.
For brevity of notation, let $\Lambda_{j, z} = \prod_{i = j}^1 W_{i, +, z}$. Using Lemma~\ref{GA:le8} yields
$$
\|\bar{v}_{x} - \frac{1}{2^d} (x - x_*)\| \leq \frac{1}{2^d} \frac{1}{16} \|x - x_*\|.
$$
It follows that
$$
\|\tilde{v}_x - \frac{1}{2^d} (x - x_*)\| = \|\bar{v}_x + \bar{q}_x - \frac{1}{2^d} (x - x_*)\| \leq \frac{1}{2^d} \frac{1}{16} \|x - x_*\| + \frac{1}{2^{d/2}} \omega.
$$

For any $x \neq 0$ and for any $v \in \partial f(x)$, by~\eqref{subdifferential-convex-hull}, there exist $c_1, c_2, \ldots, c_t \geq 0$ such that $c_1 + c_2 + \ldots + c_t = 1$ and $v = c_1 v_1 + c_2 v_2 + \ldots + c_t v_t$. It follows that
$\|v - \frac{1}{2^d} (x - x_*)\| \leq \sum_{j = 1}^t c_j \|v_j - \frac{1}{2^d} (x - x_*)\| \leq \frac{1}{2^d} \frac{1}{16} \|x - x_*\| + \frac{1}{2^{d/2}} \omega$.


\end{document}